\documentclass[a4paper,fleqn]{cas-sc}

\usepackage[authoryear,longnamesfirst]{natbib}
\setcitestyle{numbers,square,comma,sort&compress}

\usepackage[english]{babel}
\usepackage{macros}

\usepackage{amsmath, amsfonts, amssymb}
\usepackage{graphicx}
\usepackage{subcaption}
\usepackage{thmtools} 
\usepackage{thm-restate}
\usepackage{xcolor}
\usepackage{lineno}
\usepackage{verbatim}
\usepackage{booktabs}
\usepackage{cleveref}
\usepackage{orcidlink}
\usepackage{cancel}
\usepackage{overpic}
\setcitestyle{maxnames=4, minnames=3}
\def\tsc#1{\csdef{#1}{\textsc{\lowercase{#1}}\xspace}}
\tsc{WGM}
\tsc{QE}
\usepackage{xcolor}

\begin{document}
\let\WriteBookmarks\relax
\def\floatpagepagefraction{1}
\def\textpagefraction{.001}

\shortauthors{F. Naudot et al.}

\title [mode = title]{Dimensions of Power: A Systematic Guide to Power Indices for Explainable AI}
\shorttitle{Power Indices and their Principles for XAI}





\author[]{Filip Naudot}[orcid=0009-0001-8354-6265]

\cormark[1]


\ead{filipn@cs.umu.se}


\credit{Conceptualization, Methodology, Software, Validation, Formal analysis, Writing - Original Draft, Visualization, Project administration}

\affiliation[]{organization={Department of Computing Science},
            addressline={Ume{\aa} University},
            city={SE-901~87 Ume{\aa}},
            % citysep={}, % Uncomment if no comma needed between city and postcode
            country={Sweden}}

\author[]{Arunavo Ganguly}[orcid=0009-0004-9098-9904]

% Footnote of the second author

% Email id of the second author
\ead{aganguly@cs.umu.se}


\credit{Conceptualization, Methodology, Software, Validation, Formal analysis, Writing - Original Draft, Visualization}

% \affiliation[2]{organization={},
% %          citysep={}, % Uncomment if no comma needed between city and postcode

\author[]{Timotheus Kampik}[orcid=0000-0002-6458-2252]
\ead{tkampik@cs.umu.se}
\credit{Conceptualization, Methodology, Validation, Writing - Review \& Editing, Supervision}

\author[]{Vicen\c{c} Torra}[orcid=0000-0002-0368-8037]
\ead{vtorra@cs.umu.se}
\credit{Conceptualization, Validation, Writing - Review \& Editing}

\author[]{Christopher Bl{\"o}cker}[orcid=0000-0001-7881-2496]
\ead{cblocker@cs.umu.se}
\credit{Conceptualization, Methodology, Software, Validation, Formal analysis, Writing - Original Draft, Visualization}

\cortext[1]{Corresponding author}



\begin{abstract}
  Power indices, originating in cooperative game theory, quantify each player's influence on the outcome of a given game.
  Originally designed to distribute profits or costs among players and to analyse the fairness of voting systems, power indices have recently gained prominence as methods for attributing outputs of AI-based systems to inputs, thus facilitating \emph{explainability}.
  However, selecting the appropriate power index for a given explanation task is an understudied problem.
  To address this, we organise power indices along three attribution dimensions: single-player, set-based, and cardinality-based.
  For each dimension, we review the corresponding power indices, generalise existing ones where applicable, and analyse which formal principles they satisfy.
  We provide proofs for properties that are missing in the literature and show that moving to the cardinality-based setting removes player-identity information while preserving some index-level distinctions.
  Using concrete examples, we illustrate how the choice of dimension and index affects the resulting attributions in practice, and offer guidance for practitioners seeking to select a suitable power index for a given application context.
\end{abstract}


% \begin{highlights}
% \item 
% \item 
% \item 
% \end{highlights}

% Keywords

\maketitle

%% Definitions/Commands.
\newcommand{\tnote}[1]{{\color{blue}{[TK: #1]}}}
\newcommand{\fnote}[1]{{\color{orange}{[FN: #1]}}}
\newcommand{\anote}[1]{{\color{green!80!black}{[AG: #1]}}}
\newcommand{\cnote}[1]{{\color{violet}{[CB: #1]}}}

\newcommand{\mycomment}[1]{}

\Crefname{principle}{Principle}{Principle}

\newtheorem{definition}{Definition}

%%%%%%% PAPER %%%%%%%
% \linenumbers
%%%%%%%%%%%%%%%%%%%%%%%%%%%%%%%%

%%%%%%%%%%%%%%%%%%%%%%%%%%%%%%%%
\section{Introduction}
\label{section:introduction}
%%%%%%%%%%%%%%%%%%%%%%%%%%%%%%%%
Explainable AI (XAI) techniques often rely on so-called \emph{attribution functions} that quantify the effect of inputs, such as features or feature values, on outputs, such as classification decisions, or generated text.
Fundamentally, such attribution functions can be viewed as \emph{power indices}, that is, functions quantifying players' value contributions in a cooperative game.
The arguably most prominent power index is the Shapley value~\cite{shapley1953}, which relates a player's power to its average marginal contribution across all possible coalitions of players.
Since the advent of XAI research, the Shapley value and its variants have been applied to explain feature importance in machine learning models~\cite{pmlr-v119-sundararajan20b,ijcai2022p778,Chen2023, naudot2025llmshap}.
The idea is to treat each feature as a game-theoretic player and evaluate its importance based on how much it influences the model's predictions.
For example, measuring feature importance can unearth undesirable biases at the instance or model level \cite{naudot2025scalable}.

Motivated by different use cases, several power indices exist beyond the Shapley value.
For example, the Shapley-Shubik power index assesses the voting power of players in simple yes-no voting scenarios \cite{shapley-shubik1954}.
The Banzhaf power index can be understood as a generalisation of the Shapley-Shubik power index to situations where voting rights are not equally distributed \cite{Banzhaf1965}.
The Owen value refines the Shapley value by accounting for the fact that coalition formation may be constrained by a priori group structure \cite{owen1977}.
While most power indices were designed to measure the power of individual players, we also consider alternative ``player dimensions'', shifting the focus from measuring the impact of single players to sets of players and the cardinality of those sets, that is, how much the number of players influences the output regardless of individual identities.

Power indices are associated with \emph{principles}, also called \emph{axioms}, that capture useful or intuitive properties.
In XAI, these principles matter because they make explicit what kind of behaviour an attribution method guarantees, such as whether explanations ignore irrelevant features or remain stable under feature renaming.
However, to date, there is no systematic overview of power indices across player dimensions---single players, sets of players, or set cardinality---a gap we close here.
We review the background for power indices and summarise the most common ones along with their properties.
Furthermore, we introduce additional principles and provide proofs that show which power indices satisfy them.
Our work intends to serve two purposes,
the first is to collect power indices and their properties in one place, including providing proofs for properties that are so far lacking in the literature; the second is to provide a starting point for those who seek to select a power index for a specific application.

An open-source reference implementation of all surveyed power indices is---alongside documentation, tests, and all examples provided in this paper---available at \url{https://github.com/filipnaudot/powerXAI}.

%%%%%%%%%%%%%%%%%%%%%%%%%%%%%%%%
\section{Preliminaries}
\label{section:preliminaries}
%%%%%%%%%%%%%%%%%%%%%%%%%%%%%%%%
This section introduces the fundamental game-theoretic notions we use throughout the paper.
For an intuitive overview, see Figure~\ref{fig:overview-prelims}.
We work in the standard framework of cooperative games in characteristic-function form, as used by Shapley, with a finite set of players \cite{shapley1953}.

\begin{figure}[h]
    \centering
    \includegraphics[width=\linewidth]{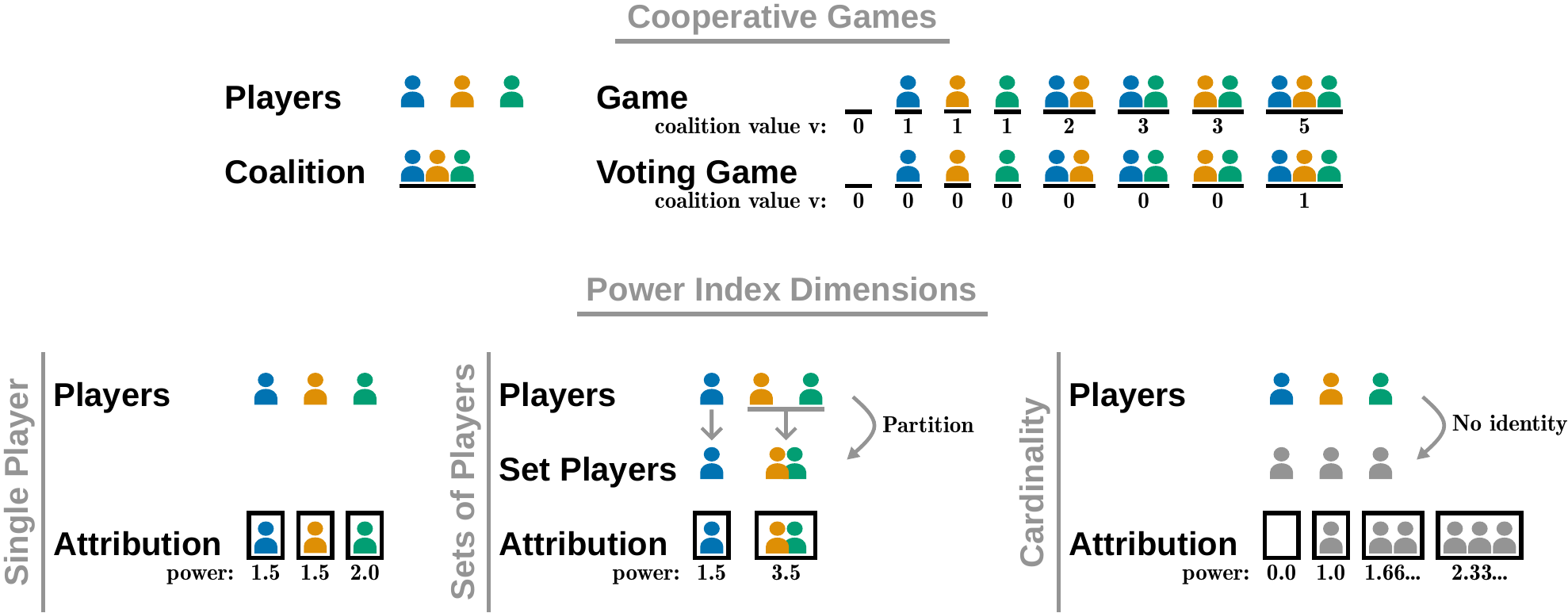}
    \vspace*{0\baselineskip}
    \caption{Illustration of the core concepts behind power indices. We depict \textbf{players} as game pieces, and \textbf{coalitions} as groups of players underlined with a bar. The \emph{grand coalition} contains all players. In cooperative \textbf{games}, the value function $v$ assigns values $v\left(S\right)$ to coalitions $S$, which we denote with a number under the coalition. In the example game shown here, each player alone ``creates'' a value of $1$; the green player's presence further increases the value by $1$ for each other player in the coalition. \textbf{Voting games} have a value of either $0$ or $1$; here we show an example voting game that requires unanimity. We consider three power-index dimensions in this work, using the value function from the game above: \textbf{Single-player} power indices assign power to individual players; here, we use the Shapley value \cite{shapley1953}. \textbf{Set-based} power indices consider player partitions and assign power to sets of players; here, we use the set-based adaptation of the Shapley value. Importantly, the orange and green players form an inseparable unit in this case. \textbf{Cardinality-based} power indices only consider the number of players, but discard their identities, which we represent with grey players; here, we use the $\Upsilon$ value \cite{Upsilon-values2024}.}
    \label{fig:overview-prelims}
\end{figure}

\topic{Players} 
Let $N$ be a finite, non-empty set.
We call the elements of $N$ \emph{players}.
These are the primitive units of interest, for example, agents or model features.

\topic{Coalition}
Given a player set $N$, a \emph{coalition} is any subset $S \subseteq N$.
Intuitively, we can interpret $S$ as a set of players that ``act together''.
The coalition containing all players $S = N$ is called the \emph{grand coalition}.

\topic{Game}
A cooperative \emph{game} on $N$ is a function $v: 2^{N} \rightarrow \mathbb{R}$, with the convention that $v(\emptyset) = 0$.
Following standard terminology, we may refer to both the function $v$ and the pair $\game$ simply as ``the game''.
For a coalition $S$, we call $v(S)$ the \emph{value} of that coalition.
Intuitively, $v(S)$ measures the game's outcome when exactly the players in $S$ ``act together''.
We deliberately leave the interpretation of $v$ open.
For a fixed $N$, the game $v$ with $v(S) = 0$ for every coalition $S \subseteq N$ is referred to as the \emph{trivial game}.

\topic{Voting games}
A \emph{voting game} on $N$ is a game $v$ 
such that for all $S, S' \subseteq N$ it holds that \emph{(i)} $v(S) \in \{ 0, 1 \}$ and \emph{(ii)} $v$ is \emph{monotone}, that is, if $S \subseteq S'$ then $v(S) \leq v(S')$.
In this setting, coalitions with $v(S) = 1$ are winning, and those with $v(S) = 0$ are losing.
\emph{Voting games} are standard models of yes/no collective decisions, such as voting rules in committees and parliaments.
In our explainability setting, we mainly consider real-valued games, that is, games that are not restricted to $\{0, 1\}$, and need not be monotone.

\topic{Power-index dimensions}
A \emph{power-index dimension} specifies the conceptual level at which power is attributed.
Formally, for a player set $N$, it specifies the domain $\mathbb{D}(N)$, and a power index $\mu$ assigns a value to each object in $\mathbb{D}\left(N\right)$, that is, $\powerindexempty(N,v)\colon \mathbb{D}(N) \mapsto \mathbb{R}$.

We distinguish three dimensions:
\begin{dimension}[Single-player dimension]\label{def:single-player-dim}
A power index $\powerindexempty$ is \emph{single-player dimension} iff the units of interest are the individual players, that is, $\mathbb{D}(N) = N$. 
Equivalently, it assigns a score $\powerindexempty_i(N,v)$ to each player $i \in N$.
\end{dimension}

\begin{dimension}[Set-based dimension]\label{def:set-based-dim}
A power index $\setpowerindexempty$ is \emph{set-based dimension} iff the units of interest are sets of players, that is, $\mathbb{D}(N) \subseteq 2^N$. 
Equivalently, it assigns a score $\setpowerindexempty_X(N,v)$ to player sets $X \in \mathbb{D}\left(N\right)$.
\end{dimension}

\begin{dimension}[Cardinality-based dimension]\label{def:cardinality-based-dim}
A power index $\cardpowerindexempty$ is \emph{cardinality-based dimension} iff it only depends on the number of considered players, but not on which exact players they are with $\mathbb{D}\left(N\right) = \left[1, \left|N\right|\right] \subset \mathbb{N}$.
Formally, for every game $(N,v)$ and all $X,Y \subseteq N$,
\[
    |X| = |Y| \; \Rightarrow \; \cardpowerindexempty_{|X|}(N,v) = \cardpowerindexempty_{|Y|}(N,v).
\]
Equivalently, it can be written as scores $\cardpowerindexempty_c(N,v)$ indexed only by $c \in \{1, \dots, |N|\}$.
\end{dimension}

\topic{Payouts} Given a game $\game$, for each player $i \in N$, player set $X \subseteq N$ and cardinality $c \in \{1, 2,\dots, n\}$, $\powerindex{i}{v}$, $\setpowerindex{X}{v}$ and $\cardpowerindex{c}{v}$ are also called \emph{payouts} w.r.t. $i, X$ and $c$, respectively.

%%%%%%%%%%%%%%%%%%%%%%%%%%%%%%%%
\section{Power Indices}
\label{section:power-indices}
%%%%%%%%%%%%%%%%%%%%%%%%%%%%%%%%
We restrict our analysis to a set of well-established power indices.
Although not exhaustive, this selection spans the main conceptual differences between attribution schemes and is therefore sufficient to support the broader conclusions of our analysis.
We present these indices for a principle-based analysis (\Cref{section:principle-analysis}), and to demonstrate how different notions of power naturally extend from single players to sets of players and cardinality-based measures.

%%%%%%%%%%%%%%%%%%%%%%%%%%%%%%%%%%%%%%%%%%%%
% Single-Player Power Indices
%%%%%%%%%%%%%%%%%%%%%%%%%%%%%%%%%%%%%%%%%%%%
\subsection{Single-Player Power Indices}
The simplest notion of contribution is the \emph{leave-one-out} contribution of player $i$, obtained by comparing the value of the grand coalition with the value obtained by removing only $i$ from it.
This is counterfactual in the sense that it asks what would happen if $i$ were absent while all other players remained present.
\begin{definition}[Counterfactual]\label{powerindex:counterfactual}
    Given a game $\game$ and a player $i \in N$, the counterfactual value for player $i$ is:
    \[
        \counterfact{i}{v} = v\left(N\right) - v\left(N \setminus \left\{i\right\}\right).
    \]
\end{definition}
This corresponds to player $i$'s marginal contribution in the single-coalition context where all other players are already present.
Moving beyond this baseline, we consider other coalition contexts as well.
The key idea of power indices is to determine marginal contributions, intuitively by asking: \emph{how much does a coalition's value change when player $i$ joins?}
A sub-question is then which coalitions we consider ``likely'' to form.
The Shapley value~\cite{shapley1953} answers the latter question by assuming that players join coalitions in uniform random order; a player's power is then its average marginal contribution at the moment they join.
\begin{definition}[Shapley Value]\label{powerindex:shapley}
    Given a game $\game$ and a player $i \in N$, the Shapley value for player $i$ is defined as:
    \begin{align*}
        \shap{i}{v}
        = \sum_{S \subseteq N \setminus \left\{i\right\}} \frac{\left|S\right|! \left(|N| - \left|S\right| - 1\right)!}{|N|!} \left(v\left(S \cup \left\{i\right\}\right) - v\left(S\right)\right)
        = \frac{1}{|N|!} \sum_{R \in \mathcal{R}} \left[ v\left(P^R_i \cup \left\{i\right\}\right) - v\left(P^R_i\right) \right].
    \end{align*}
\end{definition}
Here, $R$ denotes an ordering, or permutation, of the players, and $P_i^R$ is the set of players that precede player $i$ in the sequence defined by $R$, and $\mathcal{R}$ is the set of all orderings.
The Shapley value can be understood as a player's contribution to the grand coalition's outcome and, depending on the application, can be used to fairly allocate gains or costs based on each player's contribution.
Power indices based on marginal contributions can differ in two main ways: (i) they may differ in how the value function $v$ is interpreted or restricted, for example, when moving from general real-valued games to binary voting games; (ii) they may differ in how they weight the coalitions for which marginal contributions are evaluated.
This weighting reflects an assumption about which coalitions are more likely, or more relevant, to consider.
The Shapley value weights coalitions according to the number of player orderings that generate them, whereas the Banzhaf value~\cite{Banzhaf1965} assigns equal weight to all coalitions.
Figure~\ref{fig:shapley-banzhaf-weights} illustrates the resulting differences in weights across coalition sizes.
\begin{figure}
    \centering
    \includegraphics[width=0.5\linewidth]{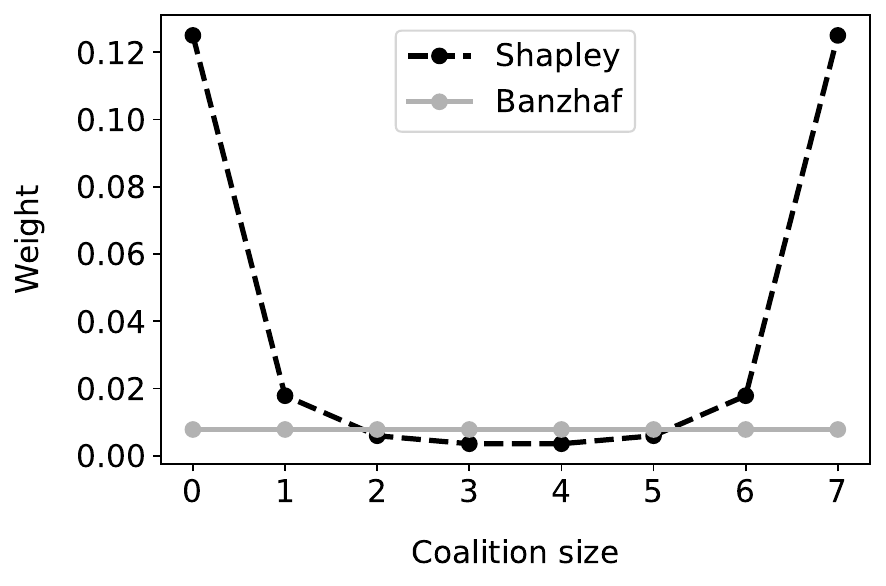}
    \caption{Coalition weights for the Shapley value and the Banzhaf index with $8$ players.
    The plot shows coalition sizes $0$ to $7$ because marginal contributions are evaluated with respect to coalitions that exclude the player under consideration.}
    \label{fig:shapley-banzhaf-weights}
\end{figure}

\begin{definition}[Banzhaf Value]\label{powerindex:banzhaf}
    Given a game $\game$ and a player $i \in N$, the Banzhaf value for player $i$ is defined as:
    \begin{equation*}
        \banzhaf{i}{v} = \frac{1}{2^{|N|-1}} \sum_{S \subseteq N \setminus \left\{i\right\}} \left( v\left(S \cup \left\{i\right\}\right) - v\left(S\right) \right).
    \end{equation*}
\end{definition}

The Owen value~\cite{owen1977} follows the Shapley value in measuring a player's power through average marginal contributions, but restricts the coalitions that may form, assuming that some coalitions would never form in practice.
When we have prior information that players tend to form specific groups, we should not treat all coalitions as equally likely.
For example, in an XAI setting, features may naturally be grouped by their source or by semantic similarity.
The core idea behind the Owen value is to compute an attribution that respects an a priori coalition structure by effectively applying the Shapley value in two stages: ``between'' and ``within'' the relevant groups defined in the a priori coalition structure.
This way, a player's contribution can be evaluated in a more realistic context.
\begin{definition}[Owen Value]\label{powerindex:owen}
    Given a player set $N$ with an a priori coalition structure $M_N = \{S^1, \ldots, S^m\}$, that is, a partition of $N$, let $S^k$ denote the block containing player $i$.
    Then, for each subset $T\subseteq M_N \setminus \{S^k\}$, we define $Q := \bigcup_{S^j \in T} S^j$,
    that is, the set of all players belonging to the blocks in $T$ (since $T$ is a coalition of blocks).
    The Owen value for player $i \in S^k$ is defined as:
    \begin{equation*}
        \owen{i}{v} = \sum_{T \subseteq M_N \setminus \{S^k\}} \sum_{K \subseteq S^k \setminus \{i\}} \frac{|K|! \left(|S^k| - |K| - 1\right)! |T|! \left(|M_N| - |T| - 1\right)!}{|S^k|! |M_N|!} \left( v\left(Q \cup K \cup \left\{i \right\}\right) - v\left(Q \cup K \right) \right).
    \end{equation*}
\end{definition}
The outer sum ranges over possible collections of groups other than $i$'s own group, while the inner sum ranges over possible subsets of $i$'s group that do not contain $i$.
Each term measures the marginal contribution of adding $i$ to the coalition formed by the players of those outside groups together with the selected members of $i$'s own group, and the weight combines the Shapley-style weights between and within groups.

%%%%%%%%%%%%%%%%%%%%%%%%%%%%%%%%%%%%%%%%%%%%
% Set-based Power Indices
%%%%%%%%%%%%%%%%%%%%%%%%%%%%%%%%%%%%%%%%%%%%
\subsection{Set-based Power Indices}
When using power indices, one may be interested in assigning power to an entire set of players rather than to each player individually.
This can be useful whenever players naturally cluster into coherent groups, such as correlated features, input modalities, or measurements from the same sensor, where single-player attribution is either uninterpretable or unnecessary.
However, the set-based power indices we introduce are not fundamentally different from single-player power indices.
The reason for introducing set-based indices separately is therefore not that they require a different conceptual foundation, but that using them correctly requires care. 
One must be explicit about which sets are treated as inseparable units, how the partition is chosen, and what question the resulting attribution is meant to answer.
Once a partition of the player set has been fixed, each block of the partition can be treated as a new ``player'', and the standard single-player indices can, in principle, be applied to the resulting relabelled game.
However, comparisons should only be made within the game induced by that same partition.
Changing the partition changes the game itself, so scores obtained under different partitions are not directly comparable, and the properties of the index need not transfer across such comparisons.
In particular, one must distinguish whether the goal is to capture interactions among players acting together as a coalition, similar to \citep{Grabisch1999}, or to measure how much a group contributes to the grand total when treated as an inseparable unit.
This choice matters because the same group can receive a high score for different reasons: either because its members create a strong joint effect, or because the group contains several individually important features.
In this paper, we adopt the latter perspective: set-based power indices apply the logic of single-player indices while treating sets of players as the basic units of attribution.
Grouping also yields a computational benefit: treating each ``block'' as a single entity reduces the number of players, thereby considerably reducing the number of coalitions we need to compute.

We start from an underlying player set $N$ and a partition $\playerpartition = \{P_1, \dots, P_{|\playerpartition|}\}$ of $N$, where each block $P_i \subseteq N$ represents a set of players that is assumed to act as an \emph{inseparable unit}, meaning, its members are only allowed to join coalitions together.
Accordingly, the resulting game is defined on the blocks, that is, $(\playerpartition, \setgame)$ rather than $\game$, where for any $S \subseteq \playerpartition$,
\[
    \setgame(S) \;:= \; v \!\left(\bigcup_{ s \in S} s \right).
\]
All set-based power indices are applied to the game $(\playerpartition, \setgame)$ and return scores for sets of players $P_i \in \playerpartition$.
Note that modifying the partition $\playerpartition$ alters the induced game $\setgame$, since it determines which coalitions can form and thereby affects the resulting attributions.

The set-based counterfactual $\setcounterfactempty$ extends the single-player counterfactual naturally by attributing power to the pre-defined sets of players.
\begin{definition}[Set-based Counterfactual]\label{powerindex:set-counterfactual}
    Given a game $(\playerpartition, \setgame)$ and a set of players $\playerset \in \playerpartition$, the set-based counterfactual value for $\playerset$ is defined as:
    \[
        \setcounterfact{\playerset}{\setgame} = \setgame\left(\playerpartition \right) - \setgame\left(\playerpartition \setminus \{\playerset\}\right).
    \]
\end{definition}

Similar to the set-based counterfactual, the set-based Shapley $\setshapempty$ and Banzhaf $\setbanzhafempty$ value extend their single-player versions by considering sets of players as the basic units of attribution when forming coalitions and computing weight factors.
\begin{definition}[Set-based Shapley Value]\label{powerindex:set-shapley}
    Given a game ($\playerpartition, \setgame)$ and a set of players $\playerset \in \playerpartition$, the set-based Shapley value for $\playerset$ is defined as:
    \begin{align*}
    \setshap{X}{\setgame} =
    \sum_{S \subseteq \playerpartition \setminus \{\playerset\}} \frac{|S|! (|\playerpartition| - |S| - 1)!}{|\playerpartition|!}\left(\setgame(S \cup \{\playerset\}) - \setgame(S)\right).
    \end{align*}
\end{definition}
\begin{definition}[Set-based Banzhaf Value]\label{powerindex:set-banzhaf}
    Given a game ($\playerpartition, \setgame)$ and a set of players $\playerset \in \playerpartition$, the set-based Banzhaf value for $\playerset$ is defined as:
    \[
        \setbanzhaf{X}{\setgame} = \frac{1}{2^{|\playerpartition|-1}} \sum_{S \subseteq \playerpartition \setminus \{\playerset\}} \left( \setgame\left(S \cup \{\playerset\}\right) - \setgame\left(S\right) \right).
    \]
\end{definition}

Even when evaluating sets of players, the Owen value's ability to respect a priori information regarding partition formation is useful.
Instead of evaluating how individuals interact between and within predetermined groups, one may want to evaluate how the groups themselves interact with other groups within a predetermined ``structure''.
\begin{definition}[Set-based Owen Value]\label{powerindex:set-owen}
    Given a game $(\playerpartition, \setgame)$ and a set-based player $X \in \playerpartition$, let $M_\playerpartition = \{S^1, \ldots, S^m\}$ be an a priori coalition structure,
    that is, a partition of $\playerpartition$.
    Let $S_X \in M_\playerpartition$ denote the block containing $X$.
    Then, for each subset $T \subseteq M_\playerpartition \setminus \{S_X\}$, define $Q := \bigcup_{S^j \in T} S^j$, that is, $Q$ is the set of all set-based players contained in the
    coalition-structure blocks in $T$.
    The Owen value for the set-based player $X$ is defined as:
    \[
        \setowen{X}{\setgame} =
        \sum_{T \subseteq M_\playerpartition \setminus \{S_X\}}
        \sum_{K \subseteq S_X \setminus \{X\}}
        \frac{
            |K|! \,
            \bigl(|S_X| - |K| - 1\bigr)! \,
            |T|! \,
            \bigl(|M_\playerpartition| - |T| - 1\bigr)!
        }{
            |S_X|! \, |M_\playerpartition|!
        }
        \left(
            \setgame\!\left(Q \cup K \cup \{X\}\right)
            -
            \setgame\!\left(Q \cup K\right)
        \right).
    \]
\end{definition}

%%%%%%%%%%%%%%%%%%%%%%%%%%%%%%%%%%%%%%%%%%%%
% Cardinality-based Power Indices
%%%%%%%%%%%%%%%%%%%%%%%%%%%%%%%%%%%%%%%%%%%%
\subsection{Cardinality-based Power Indices}\label{section:cardinality-power-indices}
The previously introduced power indices all measure the value of a specific player or a group of players.
However, one might be interested in how much the outcome typically changes when one more player is added, regardless of which player is added.
Consider a scenario in which a system must receive enough supporting evidence before approving a decision, such as flagging a transaction as legitimate or granting access to a service.
Here, the players can be interpreted as pieces of evidence or input features, and the question is not which particular feature is decisive, but how much additional evidence is typically needed for the system output to change.
Cardinality-based power indices are well-suited for this because they abstract away from player identity and treat features as interchangeable.
The Shapley value, in contrast, assumes that identities matter, replacing a highly informative feature may require a very different combination of other features to recover the same output.
Cardinality-based power indices deliberately ignore this distinction and instead ask how the expected outcome changes when one more feature is added, regardless of which feature it is.

The cardinality-based counterfactual compares the value of the grand coalition with the average value of coalitions of size $c$, regardless of which particular players are in those coalitions.
\begin{definition}[Cardinality-based Counterfactual Value]\label{powerindex:cardinality-counterfactual}
    Given a game $\game$ and a cardinality $c \in \{1, \ldots, |N|\}$, the cardinality-based counterfactual value for $c$ is defined as:
    \[
        \cardcounterfact{c}{v}
        \;=\;
        v(N) - 
        \frac{1}{\binom{|N|}{c}}
        \sum_{\substack{S\subseteq N\\
        |S|=c}} v(S).
    \]
\end{definition}

We now derive cardinality versions of the single-player Shapley, Banzhaf, and Owen power indices.
We leave the derivation of these cardinality-based alternatives to the appendix (\Cref{appx:cardinality-derivations}) and present the final derived definitions below.

The cardinality-based Shapley value considers the average change in output of $v$ when moving from coalitions of size $c - 1$ to coalitions of size $c$.
\begin{definition}[Cardinality-based Shapley Value]\label{powerindex:cardinality-shapley}
    Given a game $\game$ and a cardinality $c \in \{1, \ldots, |N|\}$, the cardinality-based Shapley value for $c$ is defined as:
    \[
        \cardshap{c}{v}
        \;=\;
        \frac{1}{\binom{|N|}{c}}
        \sum_{\substack{S \subseteq N\\ |S| = c}} v(S) - 
        \frac{1}{\binom{|N|}{c - 1}}
        \sum_{\substack{S' \subseteq N\\
        |S'|=c-1}} v(S').
    \]
\end{definition}

The cardinality-based Banzhaf value considers the same change as the cardinality-based Shapley value, namely the change in output when moving from coalitions of size $c - 1$ to coalitions of size $c$.
Unlike the cardinality-based Shapley value, however, it does not treat all coalition sizes equally.
In this sense, the relationship between the two indices is reversed compared to the single-player setting.
There, Banzhaf treats individual coalitions equally, while Shapley adjusts for coalition size, while at the cardinality-based dimension, Shapley compares coalition sizes evenly, while Banzhaf gives more influence to sizes that can be formed in more ways.
\begin{definition}[Cardinality-based Banzhaf Value]\label{powerindex:cardinality-banzhaf}
    Given a game $\game$ and a cardinality $c \in \{1, \ldots, |N|\}$, the cardinality-based Banzhaf value for $c$ is defined as:
    \[
        \cardbanzhaf{c}{v}
        \;=\;
        \frac{|N| \binom{|N|-1}{c-1}}{2^{|N|-1}}
        \left(
        \frac{1}{\binom{|N|}{c}}
        \sum_{\substack{S \subseteq N\\ |S| = c}} v(S) - 
        \frac{1}{\binom{|N|}{c-1}}
        \sum_{\substack{S' \subseteq N\\
        |S'|=c-1}} v(S')
        \right)
        \;=\; \frac{|N| \binom{|N|-1}{c-1}}{2^{|N|-1}} \cardshap{c}{v}.
    \]
\end{definition}

The cardinality-based Owen value computes how much the value changes when a coalition grows from size $c - 1$ to size $c$, but only over coalitions that the a priori structure permits, which are the whole outside blocks together with some members of the joining player's own block.
Collecting all such steps whose resulting coalition has size $c$ gives the payout for that cardinality.
\begin{definition}[Cardinality-based Owen Value]\label{powerindex:cardinality-owen}
    Given a game $\game$ and a cardinality $c \in \{1, \ldots, |\players|\}$, let $M_\players = \{S^1, \ldots, S^m\}$ be a partition of $\players$ that defines an a priori coalition structure, where $\mathcal{I} = \left\{1,\ldots,m\right\}$ are indices into $M_\players$.
    For each subset $R \subseteq \mathcal{I}$, define
    \[
        Q_R = \bigcup_{j\in R} S^j.
    \]
    Thus, $Q_R$ is the set of players contained in the blocks indexed by $R$.
    For fixed $k$ and $R$, let $m_{k,R}(t)$ denote the average marginal change obtained by increasing the number of players selected from $S^k$ from $t - 1$ to $t$, while keeping all players in $Q_R$ present (see Appendix~\ref{appendix:derivation-owen-cardinality}).
    In the cardinality-$c$ term, the local block cardinality is therefore $t = c - |Q_R|$.
    The cardinality-based Owen value for $c$ is defined as:
    \begin{equation}
        \cardowen{c}{v} = \sum_{k=1}^m \sum_{r=0}^{m-1} \frac{r!\left(m-1-r\right)!}{m!} \sum_{\substack{R \subseteq \mathcal{I} \setminus \left\{k\right\}, \left|R\right| = r \\ t = c - \left|Q_R\right|, 1 \leq t \leq |S^k|}} m_{k,R}\left(t\right).
    \end{equation}
\end{definition}

The cardinality-based power index in \Cref{powerindex:cardinality-shapley} coincides with a cardinality-based power index called the $\upsiempty$-value \cite{Upsilon-values2024}, defined via maximal chains, which are all possible coalition sequences from $\emptyset$ to $N$ in which one player is added at each step.
Proposition~2 in~\cite{Upsilon-values2024} proves that this chain-based definition reduces exactly to \Cref{powerindex:cardinality-shapley} above.
$\upsiempty$-values measure the average marginal change in payoff when moving from coalitions of size $c - 1$ to size $c$ over all maximal chains. 
As such, the $\upsiempty$-value is fundamentally a measure over coalition cardinality, not player identity.
\begin{definition}[$\upsiempty$ Value]\label{powerindex:upsilon}
    Given a set of all maximal chains $\mathcal{M}$ over the player set $N$, the $\upsiempty$ value is defined as:
    \[
    \upsi{c}{v} = \frac{1}{|\mathcal{M}(N)|}\sum_{{\mathcal{C} \in \mathcal{M}}(N)}(v(\mathcal{C}_c) - v(\mathcal{C}_{c-1})),
    \]
    where ${\cal C}_c$ denotes the $c^{th}$ member of the chain ${\cal C}$.
\end{definition}

In the context of cooperative game theory, the $\upsiempty$ value was introduced for the class of games $\game$ in which $v$ is monotone.
For our purposes, we adapt it for non-monotonic games $v$.

%%%%%%%%%%%%%%%%%%%%%%%%%%%%%%%%%%%%%%%%%%%%%%%%%%%%%%%%%%%%%%%%%%%%%%
\section{Principles}
\label{section:principles}
%%%%%%%%%%%%%%%%%%%%%%%%%%%%%%%%%%%%%%%%%%%%%%%%%%%%%%%%%%%%%%%%%%%%%%
We introduce seven principles, and group them into classical and non-classical ones.
To explain the purpose behind each principle, we provide examples either over the feature set $\{x_1,x_2,x_3\}$ or $\{x_1,x_2\}$ for various classifiers.
The classical principles capture expectations about how power should be distributed, and have historically been used to characterise indices such as the Shapley value, while the non-classical principles capture requirements that become important when power indices are used for XAI.

\subsection{Classical Principles}
The principles appear in the literature on cooperative game theory, where they provide a principle-based equivalence characterisation of power indices~\cite{Banzhaf1965,mathematical_properties_Banzhaf_value,shapley-shubik1954,shapley1953}.
We focus on efficiency, symmetry, weak anonymity, and null-player because they express basic expectations for attribution, such as, the total value should be accounted for, equivalent features should be treated alike, irrelevant features should receive no credit, and feature names should not affect the resulting explanation.

\subsubsection*{Efficiency}
The principle of efficiency requires the burden of explainability to be shared among all features,
demanding that the aggregate of the payouts equals the total value generated by the players.
From the perspective of feature engineering, the sum of the explainability scores of individual features, that is, their respective payout, must equal the total explainability prowess of the feature set. 
Formally, the principle is stated as follows.
\begin{principle}[Efficiency]
\label{principle:efficiency}
    We say that $\powerindexempty$ satisfies efficiency iff for every game $\game$ it holds that
    $$
    \sum_{i \in N} \powerindex{i}{v} = v(N).
    $$ 
\end{principle}

\subsubsection*{Symmetry}
The principle of symmetry corresponds to ``feature blindness'': explainability scores should ignore a feature's labelling and depend on its actual contribution towards explainability. 
Formally, if a value function always treats two features identically, then by the principle of symmetry, their final payouts should be identical.
\begin{principle}[Symmetry]
\label{principle:symmetry}
    We say $\powerindexempty$ satisfies symmetry iff whenever player $i,j\in N$ satisfy
    $$
    v(S \cup \{j\})
    =
    v(S\cup \{i\})
    $$
    for every $S\subseteq N\setminus\{i,j\}$ in the game $(N,v)$, it follows that $\powerindex{i}{v}=\powerindex{j}{v}.$
\end{principle}
%
% \begin{example}
%     Notice that $\kappa$ in \Cref{tab:classifier} does not distinguish between the feature $\mathbf{x_1}$ and $\mathbf{x_2}$.
%     For the sake of fairness, we pick a value function 
%     $v:2^{\{\mathbf{x_1,x_2,x_3}\}}\rightarrow \mathbb{R}$ 
%     that behaves the same whenever either of $\mathbf{x_1}$ or $\mathbf{x_2}$ are included in the coalition.
%     We formally define $v$ as:
%     \[v(S):=
%     \begin{cases}
%         0, &\text{if $S=\{\bf x_3\}$, $S=\emptyset$}
%         \\
%         1, &\text{otherwise}
%     \end{cases}
%     \]
%     Expectation is, final payouts will be identical for the features $\mathbf{x_1}$ and $\mathbf{x_2}$.
%     If we use the Shapley index, then it shows that $\shap{1}{v}=\shap{2}{v}= 0.5$, whereas $\shap{3}{v}=0$.
% \end{example}

\subsubsection*{Null Player}

The null player principle implies a zero payout towards features that do not contribute according to the value function $v$.
The closely related concept of a dummy player instead gives the player a non-zero constant value.
Formally, if the value function does not change upon the inclusion of some feature across all coalitions, then that feature is attributed a zero explainability score as per the null-player principle.
\begin{principle}[Null Player]
\label{principle:null-player}
    We say $\powerindexempty$ satisfies null player iff every player $i \in N$ who satisfies 
    $$
    v(S\cup \{i\}) = v(S)
    $$
    for every $S\subseteq N\setminus\{i\}$ in the game $\game$ receives the payout $\powerindex{i}{v} = 0$.
\end{principle}

\subsubsection*{Weak Anonymity}\label{subsection:weak-anonymity}
The principle of weak anonymity ensures that the order of features does not affect the payout of individual players.
In other words, if the features are renamed and the value function is adjusted accordingly, and each feature receives the same payout as it did before renaming.
The principle is formally defined via permutations over the set of players $N=\{1,\dots,n\}$.
A permutation $\pi:N\rightarrow N$ is a bijective function. 
Based on $\pi$, we define a function $\hat{\pi} :2^N \rightarrow 2^N$ that alters every coalition $S\subseteq N$ in the following way:
$$
\hat{\pi}(S):=\{\pi(m):m\in S\}.
$$ 
We abuse notation and denote both the functions $\pi$ and $ \hat{\pi}$, by $\pi$.
Based on that, we can derive $(N,v_\pi)$ from the original game $\game$, where $v_\pi(\hat{\pi}(S)):=v(S)$ for every coalition $S\subseteq N$. 
We also note, w.r.t. all versions of the Owen value, that a permutation acts on the whole specification of the game.
Renaming the players induces the game $(N, v_\pi)$ as above and, for indices that take an a priori coalition structure, the renamed structure $\pi(M_\players) = \{\pi(S^1), \ldots, \pi(S^m)\}$.
Renaming a feature should rename it wherever it occurs, including in the block structure.

The weak anonymity principle is defined as follows.
\begin{principle}[Weak Anonymity]
\label{principle:weak-anonymity}
     We say $\powerindexempty$ satisfies weak anonymity iff for every permutation $\pi$ over $N$, and the game $(N,v_\pi)$, it follows:
     $$
     \powerindex{\pi(i)}{v_\pi}=\powerindex{i}{v},
     $$
     for every player $i \in N$. 
\end{principle}

\subsection{Non-classical principles}
We now introduce the non-classical principles meant to answer two questions that arise while using power indices: (i) Under what circumstances can the use of power indices be considered successful? (ii) How do the more involved power indices relate to the simpler cases of explainability? 
Based on the above, the non-classical principles of success, counterfactuality, and quantitative counterfactuality are introduced.
These principles are formulated in the context of the broader explainability literature (see~\cite{Leila2022-axiomatic-foundations,KampikPYCT24-principle}) and are intended to capture aspects beyond classical principles.

%%%%%%%%%%%%%%%%%%%%%
\subsubsection*{Success}
%%%%%%%%%%%%%%%%%%%%%
The principle of success is based on the work of \citet{Leila2022-axiomatic-foundations}.
At an intuitive level, the success principle stipulates that every output of a black-box classifier can be explained with respect to the chosen explainability method. 
This stipulation is straightforward for power indices as every output can be assigned payouts via designated value functions.  
Hence, we adapt the principle at the level of payout classification. 
In other words, the success principle for power indices ensures that a power index ``successfully'' explains a black-box if for every output there is some feature that receives a non-zero payout. 
Formally, it is introduced as follows.
\begin{principle}[Success]
\label{principle:success}
    We say $\powerindexempty$ satisfies success iff every game $\game$, with $v$ not being the zero function, has a player $i \in N$ such that $\powerindex{i}{v} \neq 0$.
\end{principle}

%%%%%%%%%%%%%%%%%%%%%
\subsubsection*{Counterfactuality}
%%%%%%%%%%%%%%%%%%%%%
The principle of counterfactuality proposes that the explainability score of any given feature should be aligned with the direction of change in value between the grand coalition, containing all features, and the coalition obtained by removing that feature from the grand coalition.
The notion distinguishes between three cases: positive, zero and negative payout, which correspond to added, trivial, and negative contribution, respectively, when the feature is added to the coalition of all the other features.
\begin{principle}[Counterfactuality]
\label{principle:counterfactuality}
We say $\powerindexempty$ satisfies counterfactuality
iff for every $i \in N$ the following holds:
\begin{itemize}
    \item 
    If $\powerindex{i}{v} < 0$, then
    $v(N) < v(N \setminus \{i\})$;
    \item 
    If $\powerindex{i}{v} = 0$,
    then 
    $v(N) = v(N \setminus \{i\})$;
    \item 
    If $\powerindex{i}{v} > 0$,
    then 
    $v(N) > v(N \setminus \{i\})$;
\end{itemize}
\end{principle}

The counterfactuality principle mirrors the counterfactual power index, that is, it is motivated by the same intuition as the counterfactual power index.
It abstracts the idea that a feature's payout should track the change in value obtained when that feature is removed from the grand coalition.
In this way, the sign of the payout reflects whether the feature contributes positively, makes no difference, or contributes negatively to the overall value.

%%%%%%%%%%%%%%%%%%%%%%%%%%%%%%%%%%%%%%%%%%%%%%%%%%%%%%
\subsubsection*{Quantitative Counterfactuality}
%%%%%%%%%%%%%%%%%%%%%%%%%%%%%%%%%%%%%%%%%%%%%%%%%%%%%%
The last non-classical principle is a refinement of the principle of counterfactuality (Principle~\ref{principle:counterfactuality}), 
based on~\cite{KampikPYCT24-principle}.
In this case, the explainability burden is precisely the difference as defined in Principle~\ref{principle:counterfactuality}.
Specifically, the counterfactuality principle suggests that if $\powerindex{i}{v} > 0$ then $v(N)-v(N\setminus\{i\}) > 0$.
However, this comparison only considers positive, neutral, or negative change and is not quantified.
Quantitative counterfactuality proposes that the difference should precisely be $\powerindex{i}{v}$.
\begin{principle}[Quantitative Counterfactuality]
\label{principle:quantitative-counterfactuality}
    A power index $\powerindexempty$ satisfies the \emph{quantitative counterfactuality principle} iff for every $i \in N$, it holds that
    $$
    \powerindex{i}{v} = v(N) - v(N \setminus \{i\}).
    $$
\end{principle}
%

%%%%%%%%%%%%%%%%%%%%%%%%%%%%%%%%%%%%%%%%%%%%%%%%%%%%%%%%%%%%%%%%%%%%%
\section{Principle-based Analysis}
\label{section:principle-analysis}
%%%%%%%%%%%%%%%%%%%%%%%%%%%%%%%%%%%%%%%%%%%%%%%%%%%%%%%%%%%%%%%%%%%%%
We analyse the power indices introduced in \Cref{section:power-indices} using the classical and non-classical principles discussed in \Cref{section:principles}. 
We organise the analysis along our three attribution dimensions: based on single players, sets of players, and cardinalities.

For each dimension, we examine which principles are satisfied by the corresponding power indices, either by giving a formal proof or by presenting a counterexample.
The proofs of all results in this section are deferred to the Appendix.

We assume $N$ to be a set of players, and $\game$ to denote a fixed but arbitrary game with the function $v:N\rightarrow \mathbb{R}$.
We fix $i \in N$ to denote an arbitrary player, and $S \subseteq N$ to be an arbitrary coalition.
Lastly, $\powerindexempty$ stands for any arbitrary power-index introduced in \Cref{section:power-indices}.

%%%%%%%%%%%%%
% SINGLE & SET
%%%%%%%%%%%%%
%%%%%%%
% CLASSICAL
%%%%%%%
\begin{table*}[ht!]
    \centering
    \renewcommand{\arraystretch}{1.25}
    \begin{tabular}{ll|cccc}
        \hline
        Dimension & Power Index 
        & Efficiency 
        & Symmetry 
        & Null player 
        & Weak Anonymity \\
        \hline
        \multirow{4}{*}{Single-player}
        & Counterfactual & \xmark & \cmark & \cmark & \cmark \\
        & Shapley        & \cmark & \cmark & \cmark & \cmark \\
        & Banzhaf        & \xmark & \cmark & \cmark & \cmark \\
        & Owen           & \cmark & \xmark & \cmark & \cmark \\
        \hline
        \multirow{4}{*}{Set-based}
        & Counterfactual & \xmark & \cmark & \cmark & \cmark \\
        & Shapley        & \cmark & \cmark & \cmark & \cmark \\
        & Banzhaf        & \xmark & \cmark & \cmark & \cmark \\
        & Owen           & \cmark & \xmark & \cmark & \cmark \\
        \hline
    \end{tabular}
    \caption{Classical axiom satisfaction for single-player and set-based power indices.}
    \label{tab:axiomatic-compliance-summary-classical}
\end{table*}

%%%%%%%
% NON-CLASSICAL
%%%%%%%
\begin{table*}[ht!]
    \centering
    \renewcommand{\arraystretch}{1.25}
    \begin{tabular}{ll|ccc}
        \hline
        Dimension & Power Index
        & Success
        & Counterfactuality
        & Quantitative Counterfactuality \\
        \hline
        \multirow{4}{*}{Single-player}
        & Counterfactual & \xmark / [\cmark] & \cmark & \cmark \\
        & Shapley        & \xmark / (\cmark) & \xmark / [\cmark] & \xmark \\
        & Banzhaf        & \xmark / (\cmark) & \xmark / [\cmark] & \xmark \\
        & Owen           & \xmark / (\cmark) & \xmark / [\cmark] & \xmark \\
        \hline
        \multirow{4}{*}{Set-based}
        & Counterfactual & \xmark / [\cmark] & \cmark & \cmark \\
        & Shapley        & \xmark / (\cmark) & \xmark / [\cmark] & \xmark \\
        & Banzhaf        & \xmark / (\cmark) & \xmark / [\cmark] & \xmark \\
        & Owen           & \xmark / (\cmark) & \xmark / [\cmark] & \xmark \\
        \hline
    \end{tabular}
    \caption{
    Non-classical axiom satisfaction for single-player and set-based power indices.
    A plain symbol indicates satisfaction or failure for arbitrary games.
    Parentheses indicate results that hold when the game is monotone, and square brackets indicate results that hold when the game is strictly monotone.
    Thus, for example, \xmark / (\cmark) means the axiom fails in general but holds for monotone games, while \xmark / [\cmark] means it fails in general but holds for strictly monotone games.
    }
    \label{tab:axiomatic-compliance-summary-non-classical}
\end{table*}

%%%%%%%%%%%%%
% CARDINALITY
%%%%%%%%%%%%%
\begin{table}[ht!]
    \centering
    \renewcommand{\arraystretch}{1.25}
    \begin{tabular}{l|ccccc}
         & \multicolumn{5}{c}{Cardinality-based Indices} \\
        \cline{2-6}
        % \hline
        Principle
        & Counterfactual
        & $\boldsymbol{\upsiempty}$ 
        & Shapley
        & Banzhaf
        & Owen \\
        \hline
        Efficiency & \xmark & \cmark & \cmark & \xmark & \cmark \\
        Dummy Cardinality & \xmark & \cmark & \cmark & \xmark & \cmark \\
        Absolute Anonymity & \cmark & \cmark & \cmark & \cmark & \cmark \\
        \hline
        Success & \xmark & \xmark & \xmark & \xmark & \xmark \\
        Counterfactuality & \cmark & \xmark & \xmark & \xmark & \xmark \\
        Quantitative Counterfactuality & \cmark & \xmark & \xmark & \xmark & \xmark \\
        \hline
    \end{tabular}
    \caption{Classical and non-classical axiom satisfaction for cardinality-based power indices.}
    \label{tab:cardinality-based-axiomatic-compliance-summary}
\end{table}

%%%%%%%%%%%%%%%%%%%%%%%%%%%%%%%%
\subsection{Single Player Dimension Analysis}
\label{section:single-player-indices-analysis}
%%%%%%%%%%%%%%%%%%%%%%%%%%%%%%%%
For the single-player dimension, we study power indices where explainability scores are assigned to single players as per the chosen power index.
The main goal is to understand what each index reveals about an individual player's contribution, and which principles it respects or violates.
This helps clarify what kind of explanation each index provides at the most basic level.
The results are summarised in \Cref{tab:axiomatic-compliance-summary-classical,tab:axiomatic-compliance-summary-non-classical}.

%%%%%%%%%%%%%%%%%%%%%%%%%%%%%%%%
\subsubsection*{Counterfactual}
\label{section:single-player-counterfactual}
%%%%%%%%%%%%%%%%%%%%%%%%%%%%%%%%
We start our analysis with the counterfactual power index $\counterfactempty$ (\Cref{powerindex:counterfactual}).
Among the non-classical principles, the counterfactual power index satisfies counterfactuality and quantitative counterfactuality.
Firstly, we note that it is enough to show that quantitative counterfactuality holds, because it implies counterfactuality.
\begin{restatable}{proposition}{quanimpliesimple}
\label{proposition:quan-implies-simple}
    If a power index $\powerindexempty$ satisfies quantitative counterfactuality, then it satisfies counterfactuality.
\end{restatable}
In the light of \Cref{proposition:quan-implies-simple}, it suffices to show that the counterfactual index satisfies quantitative counterfactuality, which requires a player's payout to equal the difference between the value of the grand coalition and the value obtained by removing that player.
This is exactly how the counterfactual index is defined (\Cref{powerindex:counterfactual}).
In that sense, the result is mostly a sanity check since the principle is tailored to capture precisely the intuition built into the index.
\begin{restatable}{theorem}{counterfactualnonstandardpositive}
\label{theorem:counterfactual-non-standard-positive}
    The single-player counterfactual power index $\counterfactempty$ satisfies the quantitative counterfactuality principle, and hence the counterfactuality principle. 
\end{restatable}
However, the counterfactual index does not satisfy the success condition.
Intuitively, this is because it only considers what happens when a player is removed from the grand coalition.
As a result, it may fail to detect contributions that arise only in smaller coalitions.
In particular, a game may be non-trivial even though removing any one player from the grand coalition does not change its value.
In such cases, every player receives a payout of $0$, and success fails.
A specific instance of a value function witnessing the failure of success is as follows. 
Consider the value function below, over the feature set $\{x_1, x_2, x_3\}$.
\[  v(S):=
    \begin{cases}
        1, &\text{if $S=\{ x_3\}$ or $S=\{ x_2\}$};
        \\
        0, &\text{otherwise}.
    \end{cases}
\]
As per the counterfactuality index, $\counterfact{{x_1}}{v} = v(\{{x_1,x_2,x_3}\}) - v(\{{x_2,x_3}\}) = 0$, and similarly $\counterfact{{x_2}}{v}= 0 =\counterfact{{x_3}}{v}$.
As success requires at least one of $\counterfact{{x_1}}{v}, \counterfact{{x_2}}{v}, \counterfact{{x_3}}{v}$ to be non-zero, $v$ clearly violates success. 

For classical principles, the counterfactual index satisfies symmetry, null player, and weak anonymity.
Intuitively, this is because the index is defined as a leave-one-out comparison.
If two players affect the game in the same way under this comparison, then the counterfactual index assigns them the same payout, which is exactly what symmetry requires.
Likewise, if removing a player does not change the value, then its payout is $0$, in line with the null player principle.
For weak anonymity, merely renaming players does not change the underlying leave-one-out comparison, leaving the assigned payouts unchanged.
It is worth emphasising that these are conditional principles: whenever the relevant condition applies, the index provides the corresponding guarantee.
This is different from a principle such as efficiency, which imposes a universal constraint on the payouts in every game.
\begin{restatable}{theorem}{counterfactualpositive}
\label{theorem:counterfactual-positive}
    The single-player counterfactual index satisfies symmetry, null player, and weak anonymity.
\end{restatable}
Efficiency, on the other hand, is violated by the counterfactual index.
As a counter-example, consider the initial value function $v$ defined over ${x_1,x_2,x_3}$: 
\[  v(S):=
    \begin{cases}
        1, &\text{if $S=\{{x_1,x_2,x_3}\}$};
        \\
        0, &\text{otherwise}.
    \end{cases}
\]
\Cref{fig:efficiency-failure-counterfactual-index} plots the cumulative payout $c$ for the counterfactual index, as per \Cref{eq:cumulative-payoff}, with $\powerindexempty$ substituted by $\counterfactempty$.
Notice that the cumulative payoff exceeds the $v(\{{x_1,x_2,x_3}\})$-line at $x_3$. 
Efficiency is clearly violated, as it requires the cumulative payoff to coincide with the $v(\{{x_1,x_2,x_3}\})$-line at $x_3$.
For the sake of completeness, we provide the formula by which the cumulative payoff is calculated:
\begin{equation}
\label{eq:cumulative-payoff}
    c_{x_i}(v) = \sum_{j \leq i} \mu_{j}(v).
    % \mu_{x_1}(v) + \mu_{x_2}(v)+ \dots + \mu_{x_{i-1}}(v) + \mu_{x_i}(v), 
\end{equation}

In other words, for a specific feature $x_i$, we sum over all the other features $x_j$ whose sub-indexing is less than $i$.

%%%%%%%%%%%%%%%%%%%%%%%%%%%%%%%%%%%%%%%%
\subsubsection*{Banzhaf}
\label{section:single-player-Banzhaf}
%%%%%%%%%%%%%%%%%%%%%%%%%%%%%%%%%%%%%%%%
The Banzhaf index is closely related to the Shapley value, since both are based on averaging marginal contributions across coalitions.
The main difference is that the Banzhaf index assigns a fixed weight to each coalition, whereas the Shapley value assigns weights based on coalition size.
Intuitively, this makes the Banzhaf index simpler and less ``informed''.
One consequence of this simpler weighting scheme is that the Banzhaf index neither satisfies efficiency nor any of the non-classical principles.
In particular, success can fail, even if only in edge cases; the intuitive reasoning is as follows: the Banzhaf index averages contributions uniformly across all coalition contexts, and positive and negative effects can cancel each other out.
As a result, a non-trivial game may still assign a payout of $0$ to every player.
A specific value function witnessing this failure is given below. Again, the value function is considered over the feature set $\{x_1,x_2,x_3\}$:
\[  v(S):=
    \begin{cases}
        0, &\text{if $S=\{{x_1, x_2, x_3}\}$ or $S=\emptyset$};
        \\
        1, &\text{otherwise}.
    \end{cases}
\]
Adding each player is beneficial in some coalition contexts but harmful in others.
In this example, these two kinds of effects occur equally often: one positive contribution and one negative one, with the remaining marginal contributions equal to zero.
Because the Banzhaf index averages all coalition contexts uniformly, these effects cancel out, and the resulting payout is $\banzhaf{i}{v}=0$, for each $i \in \{{x_1,x_2,x_3}\}$.
This clearly violates success, as it requires at least one of $\banzhaf{i}{v}$ to be non-zero.
The failure of the counterfactuality principle for the Banzhaf index is also witnessed by $v$ defined above.
Since $v(N)-v(\{x_2, x_3\}) = -1$, the principle demands that $\banzhaf{1}{v}<0$, however $\banzhaf{1}{v}=0$ as calculated above.
Hence, the counterfactuality principle is violated.
Consequently, the quantitative counterpart is also violated due to \Cref{proposition:quan-implies-simple}.

For the classical principles, it is well established that the Banzhaf index satisfies all but efficiency; we refer the reader to the proofs provided by \citet{mathematical_properties_Banzhaf_value}.
\begin{theorem}[see Theorem 1 of \cite{mathematical_properties_Banzhaf_value}]
    The Banzhaf index satisfies null player, symmetry, and weak anonymity.
\end{theorem}

The Banzhaf index violates efficiency because its fixed weights do not generally ensure that the sum of all players' scores equals the value of the grand coalition.
Consequently, the equation corresponding to efficiency (Principle~\ref{principle:efficiency}) only depends on the size and value of the largest coalition, neglecting the subtleties of the value function.
We demonstrate this observation through a concrete example. Consider the following value function over the feature set $\{x_1, x_2, x_3\}$: 
\[  v(S):=
    \begin{cases}
        0, &\text{if $S=\{{x_1, x_2, x_3}\},$ or $S=\{{ x_1}\}$ or $S=\emptyset$};
        \\
        1, &\text{otherwise}.
    \end{cases}
\]
Recall that the efficiency condition states $\sum_{1\leq i\leq 3}\banzhaf{x_i}{v}=v(\{x_1, x_2, x_3\})$. 
Some simple manipulations can transform the previous equation as follows: 
$$
\sum_{1\leq i\leq 3}2^{|\{x_1, x_2, x_3\}|-1}\cdot \banzhaf{x_i}{v}=2^{|\{x_1, x_2, x_3\}|-1}\cdot v(\{x_1, x_2, x_3\}).
$$

The right-hand side of the above only depends on the size of the coalition $\{x_1, x_2, x_3\}$, and the value generated by it as per $v$. 
On the left-hand side, the value of $2^{|\{x_1,x_2, x_3\}|-1}\cdot \banzhaf{x_1}{v}$ is calculated as:  
\begin{align*}
    2^{|\{x_1, x_2, x_3\}|-1}\cdot\banzhaf{x_1}{v}:&=  v(\{x_1, x_2, x_3\})-v(\{x_2, x_3\})\\
    &+v(\{x_1, x_3\})-v(\{x_3\})+v(\{x_1, x_2\})-v(\{x_2\}) +v(\{x_1\})-v(\emptyset) = -1.
\end{align*}

By similar calculations, $2^{|\{x_1, x_2, x_3\}|-1} \cdot\banzhaf{x_2}{v} = 1 = 2^{|\{x_1, x_2, x_3\}|-1} \cdot \banzhaf{x_3}{v}$.
Together, we see that the left-hand side amounts to $\sum_{1\leq i\leq 3}2^{|\{x_1, x_2, x_3\}|-1}\cdot \banzhaf{x_i}{v}=1$, whereas the right-hand side evaluates to $2^{|\{x_1, x_2, x_3\}|-1} \cdot v(\{x_1, x_2, x_3\})=0$.
This violates efficiency.
We provide a plot of the cumulative payoff w.r.t. $v$, using \Cref{eq:cumulative-payoff}, against $v(\{x_1,x_2,x_3\})$ in \Cref{fig:efficiency-failure-banzhaf}.

\begin{figure}
    \centering
    \begin{subfigure}[b]{0.35 \linewidth}
        \includegraphics[width=\linewidth]{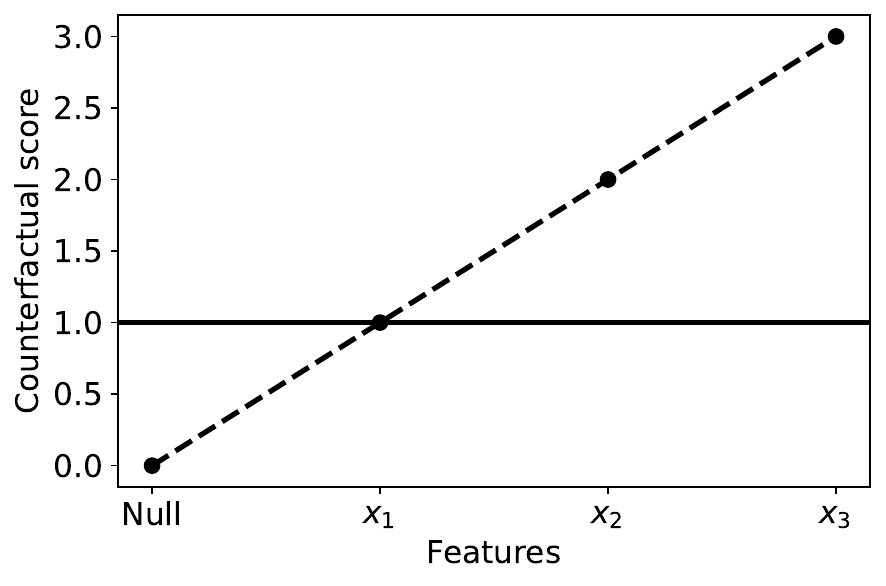}
        \caption{Efficiency failure for counterfactual index.}
        \label{fig:efficiency-failure-counterfactual-index}
    \end{subfigure}
    \hspace{0.03\linewidth}
    \begin{subfigure}[b]{0.35 \linewidth}
        \includegraphics[width=\linewidth]{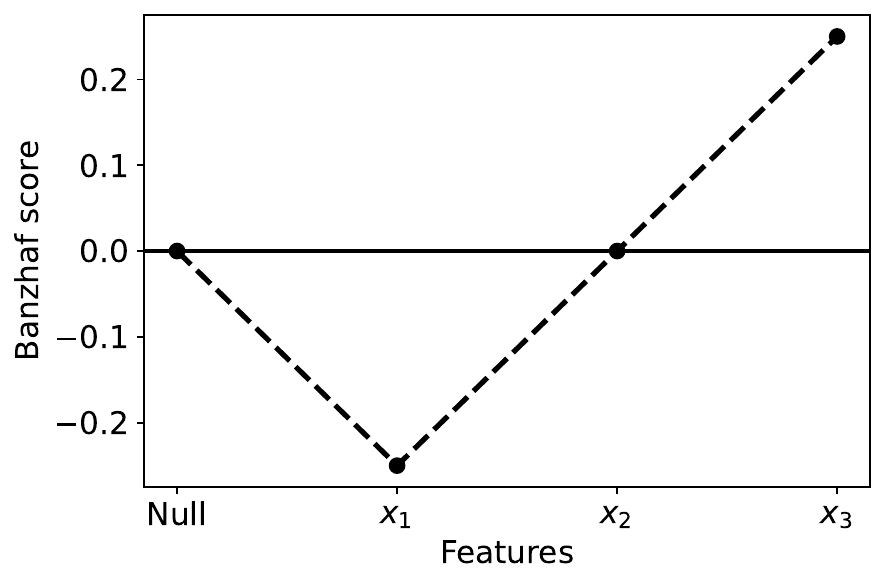}
        \caption{Efficiency failure for Banzhaf index.}
        \label{fig:efficiency-failure-banzhaf}
    \end{subfigure}
    \caption{The counterfactual and Banzhaf indices do not satisfy efficiency.
    Depending on the definition $v$, the solid line denotes the constant value $v(\{ x_1,x_2,x_3\})$; the dotted line is the cumulative power index score calculated via \Cref{eq:cumulative-payoff}. 
    In particular, the arbitrary index  $\powerindexempty$ is replaced by $\counterfactempty$ and $\banzhafempty$, respectively.}
    \label{fig:efficiency-failure-counterfactual}
\end{figure}
%
%%%%%%%%%%%%%%%%%%%%%%%%%%%%%%%%%%%%%%
\subsubsection*{Shapley}
\label{section:single-player-shapley}
%%%%%%%%%%%%%%%%%%%%%%%%%%%%%%%%%%%%%%

We turn our attention to the well-known Shapley power index. 
The classical principles for the Shapley index have been thoroughly analysed. 
In fact, together with the so-called additivity principle, they provide a principle-based characterisation of the Shapley index.
For proofs, we refer the reader to \citet{shapley1953}.
\begin{theorem}[see Section 3 in \cite{shapley1953}]
\label{theorem:shapley-standard-principles}
    The Shapley value can be characterised by the principles of efficiency, symmetry, null player, additivity and weak anonymity.
    Accordingly, the Shapley index satisfies the classical principles.
\end{theorem}

The Shapley index does not satisfy any of the non-classical principles.
These failures are due to the general nature of value functions, we only stipulate that the value function assigns $0$ to the empty feature set, whereas the values corresponding to other coalitions are arbitrary. 
This might produce non-monotonic effects, thereby cancelling the effects across various coalitions while assigning Shapley scores. 
As a specific example, consider the value function over the feature set $\{x_1, x_2, x_3\}$.
\[v(S):=
  \begin{cases}
     1, &\text{if $S=\{ x_1\}$ or $S=\{ x_2, x_3\}$};\\
     0, &\text{otherwise}.
  \end{cases}
\]

Calculation of $\shap{x_1}{v}$ yields the score of $0$, which is shown explicitly below.
Notice that the positive contribution obtained from the empty coalition is balanced by the negative contribution obtained from the coalition of size $2$, thereby amounting to $\shap{x_1}{v} = 0$.
And by similar calculations, $\shap{x_2}{v}=0=\shap{x_3}{v}$.
Therefore, the Shapley index does not satisfy success.
\begin{align*}
    \shap{x_1}{v}:= &\frac{1}{3}\cdot [v(\{x_1,x_2,x_3\})-v(\{x_2,x_3\})]+
    \frac{1}{6}\cdot\left[v(\{x_1,x_3\})-v(\{x_3\})+v(\{x_1,x_2\})-v(\{x_2\})\right]+\\
    &\frac{1}{3}\cdot [v(\{x_1\})-v(\emptyset)]=0.
\end{align*}

Additionally, as $\shap{1}{v}=0$, the principle of counterfactuality requires that $v(\{{ x_1, x_2, x_3}\})=v(\{ x_2,x_3\})$, however $v(\{ x_1, x_2, x_3\})-v(\{{x_2,x_3}\})=-1$.
Hence, the Shapley index does not satisfy the counterfactual principle, and by \Cref{proposition:quan-implies-simple}, the quantitative counterpart.

%%%%%%%%%%%%%%%%%%%%%%%%%%%%%%%%%%%%%%%%%%%%
\subsubsection*{Owen values}
%%%%%%%%%%%%%%%%%%%%%%%%%%%%%%%%%%%%%%%%%%%%
We conclude the single-player analysis with a discussion of Owen values.
In its simplest form, when players form only singleton coalitions, Owen values coincide with the Shapley values.
\Cref{theorem:simplest-case-owen} below formalises this, establishing Owen values as a generalisation of Shapley values.
\begin{restatable}{theorem}{simplestcaseowen}
\label{theorem:simplest-case-owen}
    Let $\game$ be a game with players only forming singleton partitions, that is, $M_N = \{S^1, \ldots, S^m\}$ with  $|S^j| = 1$ for every $1 \leq j \leq m$.
    It follows that $\owen{i}{v}=\shap{i}{v}$.
\end{restatable}

Because the Owen and Shapley indices coincide under singleton partitions, \Cref{theorem:simplest-case-owen} lets us transfer negative results from the Shapley value to the Owen value.
Any principle violated by the Shapley value is also violated by the Owen value.
This does not hold in the other direction, since the Owen value may violate a principle that the Shapley value satisfies when the a priori block structure is non-trivial.
In particular, Owen values fail to satisfy any of the non-classical principles.
\begin{restatable}{proposition}{owennegative}
     None of the non-classical principles is satisfied by the Owen power index.
\end{restatable}

Likewise, Owen values satisfy all of the classical principles except symmetry.
The formal proofs of these results follow arguments similar to those in the proofs of the Shapley values and we refer the reader to the proof provided by \citet{owen1977}.
\begin{theorem}[see Section 4 in \cite{owen1977}]\label{theorem:owen-classical-principles}
    The principles of efficiency, null player and weak anonymity are satisfied by the Owen power index.
\end{theorem}

Symmetry fails because the a priori structure restricts which coalitions can form, so two players whom $v$ treats identically may still be evaluated in different coalition contexts when they belong to different blocks.
\begin{restatable}{theorem}{owenSymmetryNegativeProof}
\label{theorem:owen-symmetry-principle-negative}
    The Owen value violates the symmetry principle.
\end{restatable}

The only constraint on the valuation function is that the empty coalition is always assigned the value $0$. 
These restrictions are not enough to guarantee non-classical principles, which fail to hold Banzhaf, Shapley and Owen values.
To try and salvage some of the non-classical principles, we briefly discuss the class of \emph{monotonic} and \emph{strictly monotonic} games.
A game $(N,v)$ is said to be monotonic, or strictly monotonic, iff the value function satisfies
\begin{align*}
    A\subseteq B \Rightarrow v(A)\leq v(B),  &\tag{monotone}
    \\
    A\subsetneq B \Rightarrow v(A)< v(B). &\tag{strictly-monotone}
\end{align*}
Given such constraints, the non-classical principles can be salvaged, as shown in \Cref{theorem:non-stanadard-principles-of-monotone-games}.
\begin{restatable}{theorem}{nonstanadardprinciplesofmonotonegames}
\label{theorem:non-stanadard-principles-of-monotone-games}
    If $\game$ is a game, then
    \begin{enumerate}
        \item[i.]
        If $\game$ is monotone, then the Banzhaf value $\banzhafempty$, the Shapley value $\shapempty$, and the Owen value $\owenempty$ satisfy success.
        \item[ii.]
        If $\game$ is strictly monotone, then Banzhaf value $\banzhafempty$, the Shapley value $\shapempty$ and the Owen value $\owenempty$ satisfy counterfactuality.
    \end{enumerate} 
\end{restatable}
Observe that the counterfactual index does not satisfy success in the case of monotone games. 
However, it is straight-forward to notice that the principle is satisfied in the case of strictly monotone games.

Although success is not satisfied by any of the indices in the case of non-monotonic games, this negative result is still informative.
It highlights that power index-based explanations measure changes in the value function according to a chosen aggregation scheme, but they do not guarantee that some player receives non-zero information in every environment.
If all relevant changes occur in contexts that an index ignores (e.g., certain coalitions), or if positive and negative effects cancel out under the index's weighting scheme, then all payouts may be zero even though the game itself is not trivial.
This is useful from an XAI perspective because it makes explicit that a power index only explains the kinds of changes it is designed to detect.
The counterfactuality and quantitative counterfactuality results make the same point from another angle.
The intuitive idea that the sign of a player's attribution should match the effect of removing that player from the grand coalition is not guaranteed in general, except by the counterfactual index, which is defined precisely to capture that comparison.

%%%%%%%%%%%%%%%%%%%%%%%%%%%%%%%%
\subsection{Set-based Player Dimension Analysis}
\label{section:set-player-indices-analysis}
%%%%%%%%%%%%%%%%%%%%%%%%%%%%%%%%
We move from individual players to sets of players as the basic units of attribution.
Accordingly, we study the set-based variants of the indices introduced earlier, namely the set-based counterfactual $\setcounterfactempty$, set-based Shapley $\setshapempty$, set-based Banzhaf $\setbanzhafempty$, and set-based Owen $\setowenempty$.
More generally, we write $\setpowerindexempty$ for an arbitrary set-based power index.
Since these indices are defined by applying the same underlying ideas at the level of sets rather than single players, the analysis closely mirrors the single-player case.
\begin{restatable}{theorem}{setbasedanalysis}
\label{theorem:set-based-analysis}
    For any power index $\powerindexempty$ listed in \Cref{tab:axiomatic-compliance-summary-classical,tab:axiomatic-compliance-summary-non-classical}, $\powerindexempty$ and $\setpowerindexempty$ satisfy the same principles.
\end{restatable}

\Cref{theorem:set-based-analysis} substantially reduces the burden for the set-based player analysis. Accordingly, whether a set-based power index satisfies a principle depends on whether its single-player counterpart satisfies the principle.
We summarise this information in \Cref{tab:axiomatic-compliance-summary-classical,tab:axiomatic-compliance-summary-non-classical}.

%%%%%%%%%%%%%%%%%%%%%%%%%%%%%%%%%%%%%%%%%%%%%%%%%%%%%%%%%%%%%
\subsection{Cardinality-based Player Dimension Analysis}
\label{section:cardinality-player-indices-analysis}
%%%%%%%%%%%%%%%%%%%%%%%%%%%%%%%%%%%%%%%%%%%%%%%%%%%%%%%%%%%%%

The key change from the single-player and set-based settings is that the object of attribution is no longer a particular player or set of players.
Instead, power is assigned to a coalition of size $c$.
In other words, these indices tell us how much explanatory value is associated with having $c$-many features, rather than with any specific feature or feature set.
This affects the design of desirable principles.
To motivate the need for new principles, we use the following example: 
\begin{table}
\setlength{\tabcolsep}{10pt}
\renewcommand{\arraystretch}{1.5}
    \centering
    \begin{tabular}{|c|c|c|c|}
    \hline
       $\mathbf{x_1}$  & $\mathbf{x_2}$ & $\mathbf{x_3}$ & ${\kappa}$ \\
       \hline
        0 & 0 & 0 & 0 \\
        0 & 0 & 1 & 0 \\
        0 & 1 & 0 & 1 \\
        0 & 1 & 1 & 1 \\
        1 & 0 & 0 & 0 \\
        1 & 0 & 1 & 0 \\
        1 & 1 & 0 & 0 \\
        1 & 1 & 1 & 0 \\
        \hline
    \end{tabular}
    \caption{Table capturing the behaviour of $\kappa$ over the features ${{x}_1, {x}_2, {x}_3}$.
    If ${x_1}$ is $1$ then $\kappa$ returns $0$, disregarding the two other features.
    $\kappa$ returns $1$ precisely when ${x_1}$ is $0$, and ${x_2}$ is $1$.
    As a logical formula, $\kappa$ can be obtained as $\neg x_1 \wedge x_2$, that is, $\kappa$ does not depend on $x_3$.
    }
    \label{tab:classifier}
\end{table}
\[
v(S):=
\begin{cases}
    1, &\text{if $S=\{ x_1\}$ or $S=\{ x_2\}$}
    \\
    0, &\text{otherwise}.
\end{cases}
\] 
The value is assigned as per the classifier $\kappa$ from \Cref{tab:classifier}, where $\upsi{2}{v} = -0.67$ but $\shap{2}{v} = 0.167$.
This mismatch can be explained as follows: $\upsiempty$ accounts for the contribution of feature subsets of $\{ x_1,x_2,x_3\}$ that are of cardinality 2, namely $\{x_1,x_2\}, \{x_2, x_3\}$ and $\{x_1, x_3\}$. 
The Shapley index, on the other hand, only considers the contribution of the feature $ x_2$. 
Hence, we need to adapt the principles from \Cref{section:principles} to match the intuitions behind the class of cardinality-based power indices.

First, we focus on the classical principles, starting with \emph{efficiency}.
The idea underpinning the principle remains unchanged, except that the intuition is applied to a particular cardinality instead of individual features.
That is, when explainability is understood solely in terms of the power attributed to individual cardinalities, efficiency is understood as the aggregate of the power of each individual cardinality.
Formally, it is expressed as:
\begin{principle}[Cardinality-based Efficiency]\label{principle:cardinality-efficiency}
     ${v}(N)=\sum_{1\leq i\leq |N|}\cardpowerindex{i}{v}$,
     for any coalitional game $\game$.
\end{principle}

Analogous to weak anonymity, we propose the principle of \emph{absolute anonymity}.
The property was first proposed in the context of $\upsiempty$-values~\cite{Upsilon-values2024} to describe a version of permutation invariance shown by cardinality-based indices. 
Intuitively, absolute anonymity implies that permutations of player indices do not affect the cardinality-based power indices.
Indeed, this principle closely aligns with cardinality-based indices, which are feature-blind in that they do not consider feature identities.
\begin{principle}[Absolute Anonymity]\label{principle:absolute-anonymity-cardinality}
    For any game $\game$, the power index $\cardpowerindex{i}{v}$ satisfies absolute anonymity iff for any permutation $\pi$, the following holds:
    $$
    \cardpowerindex{i}{\pi_v}=
    \cardpowerindex{i}{v}.
    $$
\end{principle}
As discussed when introducing the weak anonymity principle, a permutation acts on the whole specification of the game, so the a priori block structure is renamed along with the players.

Analogous to absolute anonymity, we introduce dummy cardinality to replace the null player principle.
Intuitively, dummy cardinality arises whenever all possible extensions to a specific cardinality yield the same result.
In other words, all possible ways of forming a particular cardinality, and all possible ways of extending it by 1, have the same effect.
\begin{definition}
    Given a game $\game$, a cardinality $i$ is a dummy cardinality iff there is a constant $h_i$ such that, for every player $x\in N$, and $S \subseteq N \setminus\{x\}$ with $|S \cup \{x\}| = i$, it holds $v(S\cup \{x\})=v(S)+h_i$. 
\end{definition}
Consequently, a cardinality-based power index satisfies the \emph{dummy cardinality principle} if the payouts for dummy cardinalities are adjusted accordingly.
Specifically, dummy cardinalities get assigned the constant value generated during the extension described above.
\begin{principle}[Dummy Cardinality]\label{principle:dummy-cardinality}
     Given a game $\game$, a power index $\cardpowerindex{i}{v}$ satisfies dummy cardinality iff for every dummy cardinality $i$ it holds that $\cardpowerindex{i}{v}=h_i$.
\end{principle}
Towards the non-classical principles for cardinality, we start with the cardinality-based success principle. 
The satisfaction of this principle requires at least one of the cardinalities having a non-zero payout score.
The formal definition of the principle remains unchanged. 
\begin{principle}[Cardinality-based Success]
    \label{principle:card-success}
    Given a game $\game$, a power index $\cardpowerindexempty$ satisfies (cardinality-based) success iff every game $\game$ with $v$ not being the zero function, has a cardinality $c$ such that $\cardpowerindex{c}{v}\neq 0$.
\end{principle}

When viewed from the cardinality point of view, $\cardpowerindexempty$ satisfies success iff at least one of $\cardpowerindex{i}{v}$ is non-zero, for some cardinality $i$.
The principle of counterfactuality and its quantitative counterpart are re-defined for the cardinality setting. 
Recall that the motivation behind the single-player and set-based counterfactuality principles was to assess the value of a particular feature by constructing a scenario where the feature is removed from the feature set and considering its difference from the grand coalition.
However, we need to reinterpret this idea in terms of cardinality, which we do by introducing the following formalism.
\begin{principle}[Cardinality-based Counterfactuality]\label{principle:cardinality-counterfactuality}
    Given a game $\game$, a power index $\cardpowerindexempty$ satisfies (cardinality-based) counterfactuality iff 
    for each game $\game$ and cardinality $i$, we have that
    \begin{itemize}
        \item 
        If $\cardpowerindex{i}{v} < 0$, then $v(N)
        < 
        \frac{1}{\binom{|N|}{i}}\left(\sum_{|S|=i}v(S)\right)$;
        \item 
        If $\cardpowerindex{i}{v} = 0$,
        then $v(N)
        = 
        \frac{1}{\binom{|N|}{i}}\left(\sum_{|S|=i}v(S)\right)$;
        \item 
        If $\cardpowerindex{i}{v} > 0$,
        then $v(N)
        > 
        \frac{1}{\binom{|N|}{i}}\left(\sum_{|S|=i}v(S)\right)$.
    \end{itemize}
\end{principle}
The quantitative variant of cardinality-based counterfactuality follows analogously.
\begin{principle}[Cardinality-based Quantitative Counterfactuality]\label{principle:cardinality-quantitative-counterfactuality}
    Given a game $\game$, a power index $\cardpowerindexempty$ satisfies (cardinality-based) quantitative counterfactuality iff 
    $$
    \cardpowerindex{i}{v}=v(N)-\left(\frac{1}{\binom{|N|}{i}}\sum_{|S|=i}v(S)\right),
    $$
    for every game $(N,v)$.
\end{principle}

As in the single-player case, the quantitative principle is the stronger of the two.
\begin{restatable}{theorem}{cardqunatitativetonormal}
\label{theorem:card-quantitative-to-normal}  
    If a power index $\cardpowerindex{i}{v}$ satisfies quantitative counterfactuality, then it also satisfies counterfactuality.  
\end{restatable}

Notice that the cardinality-based counterfactuality principle and the cardinality-based counterfactual power index have a considerable overlap in their respective definitions.
This is in accordance with the discussion that followed the single-player-based counterfactual principle (see Principle~\ref{principle:counterfactuality}).

We proceed to the analysis of cardinality-based power indices in light of the principles introduced above.

%%%%%%%%%%%%%%%%%%%%%%%%%%%%%%%%%%
\subsubsection*{$\upsiempty$-value}
\label{section:upsilon}
%%%%%%%%%%%%%%%%%%%%%%%%%%%%%%%%%%
%%%
% Upsilon (and Shapley) Analysis
%%%
The $\upsiempty$ value satisfies all of the standard cardinality-based principles;
\Cref{theorem:upsilon-values-standard} highlights these results.
\begin{theorem}[see \cite{Upsilon-values2024}]
\label{theorem:upsilon-values-standard}
    The $\upsiempty$ value satisfies the principles of efficiency, dummy cardinality, and absolute anonymity.
\end{theorem}

Towards the non-classical principles, the $\upsiempty$ value does not satisfy cardinality-based success (Principle~\ref{principle:card-success}), even when excluding value functions that always return $0$.
Intuitively, this is because the $\upsiempty$ value averages the change in value when moving from one coalition size to the next.
As a result, positive and negative changes across different coalition sizes can cancel each other out.
In such cases, the game is still non-trivial, but every cardinality-based payout is $0$, so success fails.
The following example illustrates this phenomenon.
\begin{example}[Success failure in $\upsiempty$ value]
\label{ex:success-failure-upsilon-value}
    Consider a game over the feature set $\{x_1, x_2\}$, where $v$ is defined as follows:
    \begin{equation*}
        v(\{{x_1}\})=1, \qquad v(\{{x_2}\})=-1, \qquad v(\{{x_1, x_2}\})=0.
    \end{equation*}
    Calculating the $\upsiempty$ values yield:
    \begin{equation}
    \label{eq:upsilon-value-calc}
        \begin{split}
            \upsiempty_1(v)&=\frac{1}{2!}[v(\{{x_1}\})+v(\{{x_2}\})]-v(\emptyset)=0,
            \\
            \upsiempty_2(v)&=v(\{{x_1,x_2}\})-\frac{1}{2!}[v(\{{x_1}\})+v(\{{x_2}\})]=0.
        \end{split}
    \end{equation}
    Success demands that either $\upsiempty_1(v)$ or $\upsiempty_2(v)$ should be non-zero, thereby contradicting the above.
    Hence, success is not satisfied by $\upsiempty$ values.  
\end{example}

As we show using \Cref{ex:counterfactuality-failure-upsilon-value}, cardinality-based counterfactuality is not satisfied by $\Upsilon$-values.
By \Cref{theorem:card-quantitative-to-normal}, the quantitative counterpart is also violated.
\begin{example}[(Quantitative) Counterfactuality failure in $\upsiempty$ value]
\label{ex:counterfactuality-failure-upsilon-value}
    Consider the value function $v$ over the feature set $\{ x_1, x_2, x_3\}$ defined as follows.
    \[
    v(S):=
    \begin{cases}
       1; &\text{if $S=\{{ x_1}\}$ or $S=\{{ x_1, x_2, x_3}\}$} ,
       \\
       -1; &\text{if $S=\{ x_2\}$},
       \\
       0; &\text{otherwise}.
    \end{cases}
    \]    
    Note that $\upsi{1}{v}=0$, and, hence, if the cardinality-based counterfactuality was satisfied, then it would imply $v(N)-\frac{1}{\binom{|N|}{1}}\left(\sum_{|S|=1}v(S)\right)=0$.
    However, $v(N)-\frac{1}{\binom{|N|}{1}}\left(\sum_{|S|=1}v(S)\right)=1>0$, contradicting the implication. 
    Hence, the cardinality-based counterfactuality principle is not satisfied by $\upsiempty$.
\end{example}

By Proposition~2 in~\cite{Upsilon-values2024}, the cardinality-based Shapley value (\Cref{powerindex:cardinality-shapley}) coincides with the $\upsiempty$ value.
Hence, $\cardshapempty$ satisfies exactly the same principles as $\upsiempty$ under our definitions.

%%%%%%%%%%%%%%%%%%%%%%%%%%%%%%%%%%
\subsubsection*{Cardinality-based Banzhaf value}
\label{section:analysis-cardinality-banzhaf}
%%%%%%%%%%%%%%%%%%%%%%%%%%%%%%%%%%
%%%
% Cardinality Banzhaf Analysis
%%%
Unlike the $\upsiempty$-value, the cardinality-based Banzhaf value $\cardbanzhafempty$ does not satisfy all of the classical principles.
This is conceptually consistent with the behaviour at the single-player dimension compared to the Shapley value.
%%%%%%%%%%%%%%%%%
% Card. Banzhaf violates efficiency
%%%%%%%%%%%%%%%%%
\begin{restatable}{theorem}{cardBanzhafNegativeEfficiency}
    \label{theorem:cardinality-banzhaf-value-violates-efficiency}
    The $\cardbanzhafempty$ value violates the efficiency principle.
\end{restatable}

%%%%%%%%%%%%%%%%%
% Card. Banzhaf violates dummy cardinality
%%%%%%%%%%%%%%%%%
Intuitively, the cardinality-based Banzhaf value detects the same local change between coalition sizes as the cardinality-based Shapley value, but it does not report this change directly.
Instead, it multiplies the change by a cardinality-dependent weight, which is precisely why dummy cardinality can fail.
\begin{restatable}{theorem}{cardBanzhafNegativeDummyCardinality}
    \label{theorem:cardinality-banzhaf-value-violates-dummy-cardinality}
    The $\cardbanzhafempty$ value violates dummy cardinality.
\end{restatable}

%%%%%%%%%%%%%%%%%
% Card. Banzhaf satisfies absolute anonymity
%%%%%%%%%%%%%%%%%
Absolute anonymity holds because the cardinality-based Banzhaf value ignores player identities altogether.
Renaming the players can change which coalitions have which names, but not how many coalitions of each size exist or what average change between sizes is measured.
\begin{restatable}{theorem}{cardBanzhafPositiveAbsoluteAnonymity}
    \label{theorem:cardinality-banzhaf-value-satisfies-absolute-anonymity}
    The $\cardbanzhafempty$ value satisfies absolute anonymity.
\end{restatable}

%%%%%%%%%%%%%%%%%
% Card. Banzhaf violates success
%%%%%%%%%%%%%%%%%
For the non-classical principles, the cardinality Banzhaf value does not satisfy cardinality-based success (Principle~\ref{principle:card-success}), even when excluding value functions that always return $0$.
Intuitively, this is because $\cardbanzhafempty$ compares average values between adjacent cardinalities, as $\upsiempty$ does, but additionally rescales these differences by a Banzhaf weight.
Thus, if the average difference between adjacent cardinalities is $0$, the corresponding payout remains $0$.
\begin{restatable}{theorem}{cardBanzhafNegativeSuccess}
    \label{theorem:cardinality-banzhaf-value-violates-success}
    The $\cardbanzhafempty$ value violates the success principle.
\end{restatable}

%%%%%%%%%%%%%%%%%
% Card. Banzhaf violates (quantitative) counterfactuality
%%%%%%%%%%%%%%%%%
The $\cardbanzhafempty$ index, just like the $\upsiempty$, violates the cardinality-based counterfactuality principle and
by \Cref{theorem:card-quantitative-to-normal}, the quantitative counterpart.
\begin{restatable}{theorem}{cardBanzhafNegativeCounterfactuality}
    \label{theorem:cardinality-banzhaf-value-violates-counterfactuality}
    The $\cardbanzhafempty$ value violates the Counterfactuality and Quantitative Counterfactuality principles.
\end{restatable}

%%%%%%%%%%%%%%%%%%%%%%%%%%%%%%%%%%
\subsubsection*{Cardinality-based Owen value}
\label{section:analysis-cardinality-owen}
%%%%%%%%%%%%%%%%%%%%%%%%%%%%%%%%%%
%%%
% Cardinality Owen Analysis
%%%
We now turn to the cardinality-based Owen value $\cardowenempty$.
Unlike the cardinality-based Banzhaf value, it satisfies all classical cardinality-based principles, but, like every other index considered here, it fails the non-classical ones.
%%%%%%%%%%%%%%%%%
% Card. Owen satisfies efficiency
%%%%%%%%%%%%%%%%%
Efficiency follows from how the index is constructed.
The cardinality-based Owen value re-bundles the terms of the ordinary Owen value according to the size of the coalition that is formed, so no term is lost or counted twice.
\begin{restatable}{theorem}{cardOwenEfficiency}\label{theorem:cardinality-owen-value-efficiency}
    The $\cardowenempty$ value satisfies the efficiency principle.
\end{restatable}

%%%%%%%%%%%%%%%%%
% Card. Owen satisfies dummy cardinality
%%%%%%%%%%%%%%%%%
When a cardinality is dummy, every step into that size adds the same constant.
The weights collected at a given cardinality sum to one, so the index returns that constant unchanged.
\begin{restatable}{theorem}{cardOwenDummy}\label{theorem:cardinality-owen-value-dummy}
    The $\cardowenempty$ value satisfies the dummy cardinality principle.
\end{restatable}

%%%%%%%%%%%%%%%%%
% Card. Owen satisfies absolute anonymity
%%%%%%%%%%%%%%%%%
Renaming the players also renames the a priori blocks, which leaves the block sizes and the weights untouched. 
Only the labels change, so the payouts stay the same.
\begin{restatable}{theorem}{cardOwenAbsoluteAnonymity}\label{theorem:cardinality-owen-value-absolute-anonymity}
    The $\cardowenempty$ value satisfies the absolute-anonymity principle.
\end{restatable}
Note that $\cardowenempty$ satisfies properties that carry the same names as the ones used to characterise $\upsiempty$-values in Theorem~1 of \cite{Upsilon-values2024}, which would make the two indices coincide.
However, they do not (see \Cref{table:card-owen-shap-comparison}), because in that paper the permutation acts on the game alone, whereas in our work we say that it also renames the a priori block structure, as ignoring it when renaming would not be meaningful.

%%%%%%%%%%%%%%%%%
% Card. Owen violates success.
%%%%%%%%%%%%%%%%%
The non-classical principles fail.
For success, similarly to $\cardshapempty$ and $\upsiempty$ positive and negative changes across coalition sizes can cancel out so that every payout is zero.
\begin{restatable}{theorem}{cardOwenSuccess}\label{theorem:cardinality-owen-value-success}
    The $\cardowenempty$ value violates the success principle.
\end{restatable}

%%%%%%%%%%%%%%%%%
% Card. Owen violates
% (Quantitative) Counterfactuality.
%%%%%%%%%%%%%%%%%
For counterfactuality and quantitative counterfactuality, the violation is essentially trivial from the definition as we have see previously.
\begin{restatable}{theorem}{cardOwenCounterfactuality}\label{theorem:cardinality-owen-value-counterfactuality}
    The $\cardowenempty$ value violates the Counterfactuality and Quantitative Counterfactuality principles.
\end{restatable}

%%%%%%%%%%%%%%%%%%%%%%%%%%%%%%%%%%%%%%%%%%
\subsubsection*{Cardinality-based Counterfactual Index}
\label{section:card-based-countrfactual}
%%%%%%%%%%%%%%%%%%%%%%%%%%%%%%%%%%%%%%%%%%%%
%%%
% Cardinality Counterfactual Analysis
%%%
As in the single-player case, the behaviour of the cardinality-based counterfactual index $\cardcounterfactempty$ with respect to the counterfactuality principles is largely due to its duality with the cardinality-based counterfactual principle: it compares the value of the grand coalition with the average value of coalitions of a fixed size.
For this reason, its relationship with the cardinality-based counterfactuality principles is intuitively straightforward.
From this, it follows that the cardinality-based counterfactual satisfies quantitative counterfactuality and the case of counterfactuality is handled by \Cref{theorem:card-quantitative-to-normal}. 
However, the task is trivial since we have that $\cardcounterfact{i}{v}=v(N)-\left(\frac{1}{\binom{|N|}{i}}\sum_{|S|=i}v(S)\right)$, which coincides with the requirement of quantitative counter-factuality.
\begin{theorem}
    The cardinality-based counterfactual index $\cardcounterfactempty$ satisfies the quantitative counterfactuality principle, and hence counterfactuality. $\hfill \square$
\end{theorem}

In general, the principle of success is not satisfied by the cardinality-based counterfactual index.
Violating cases arise when the contributions pertaining to a cardinality cancel each other out while calculating the average.   
One instance of this effect is shown in the Example below.
\begin{example}
\label{ex:success-violation-card-counterfact}
    Consider the example of the game $\game$ where $N=\{x_1,x_2\}$ and the following value function:
    \begin{equation*}
        v(\emptyset)=0 \qquad v(\{{x_1}\})=1 \qquad v(\{{x_2}\})=-1 \qquad v(\{{x_1,x_2}\})=0
    \end{equation*}
    Note that $\cardcounterfact{2}{v} = 0$, since $v(S)=0$ for the single $S$ of size $2$.
    Finally,
    $$
    \cardcounterfact{1}{v}=v(\{{x_1, x_2}\})-\frac{1}{\binom{2}{1}}\bigl(v(\{{x_2}\})+v(\{{x_1}\})\bigr) = 0.
    $$
    For cardinality-based counterfactual index to satisfy success, every game must have a cardinality with a non-zero payout.
    Hence, success is clearly not satisfied as this game requires at least one of $\cardcounterfact{1}{v}$ or $ \cardcounterfact{2}{v}$ to be non-zero.
\end{example}
The cardinality-based counterfactual index violates the principle of efficiency. 
In cases where strictly smaller subsets are assigned a higher value than the full feature set, according to the value function, an imbalance arises between the left- and right-hand sides of the efficiency equation.
To illustrate a concrete instance of the above point, consider a variation of \Cref{ex:success-violation-card-counterfact} below.
\begin{example}
\label{ex:efficiency-violation-card-counterfact}
    Consider the example of the game $\game$ where $N=\{x_1,x_2\}$ and the value function $v$ is binary, returning $1$ iff the feature coalition is of size $1$. 
    Formally, it is defined in the following manner:
    \begin{equation*}
        v(\emptyset)=0 \qquad v(\{{x_1}\})=1 \qquad v(\{{x_2}\})=1 \qquad v(\{{x_1,x_2}\})=0
    \end{equation*}
    Note that $\cardcounterfact{2}{v}=0$, since $v(S)=0$ for the single $S$ of size $2$.
    Finally,
    $$
    \cardcounterfact{1}{v}=v(\{{x_1,x_2}\})-\frac{1}{\binom{2}{1}}\bigl(v(\{{x_2}\})+v(\{{x_1}\})\bigr)=-1.
    $$
    Efficiency requires that $v(\{{x_1,x_2}\}) = \cardcounterfact{1}{v}+\cardcounterfact{2}{v}$ holds for the cardinality-based counterfactual index. 
    However, 
    $$
    0 \neq \cardcounterfact{1}{v}+\cardcounterfact{2}{v} = -1.
    $$
    Hence, the index does not satisfy efficiency. 
\end{example}

The cardinality-based counterfactual index also violates the dummy cardinality principle.
The cardinality-based counterfactual index considers the difference between the grand coalition and the average value of all coalitions of a given cardinality, rather than directly measuring the marginal change when the coalition cardinality increases by one.
The following counterexample is based on \Cref{ex:efficiency-violation-card-counterfact}.
In that example, cardinality 1 is a dummy cardinality with $h_1 = 1$ since $v(\{x\}) = v(\emptyset)+1$ for every $x \in \{x_1, x_2\}$.
If the cardinality-based counterfactual index satisfied dummy cardinality, it would follow that $\cardcounterfact{1}{v} = 1$.
However, as calculated above, $\cardcounterfact{1}{v} = -1$.
Hence, the cardinality-based counterfactual index $\cardcounterfactempty$ violates the dummy cardinality principle.

To conclude the analysis, we observe that the cardinality-based counterfactual index satisfies absolute anonymity.
This follows from its definition since it only depends on the grand coalition and the average value of coalitions of a fixed size, not on the identity of the players themselves.
Therefore, renaming players does not affect the resulting payouts.
\begin{restatable}{theorem}{absoluteanonymitycardcounterfactuality}
\label{theorem:absolute-anonymity-card-counterfactuality}
    The cardinality-based counterfactual index $\cardcounterfactempty$ satisfies absolute anonymity.  
\end{restatable}
%

%%%%%%%%%%%%%%%%%%%%%
\section{Illustrative Examples}
\label{section:examples}
%%%%%%%%%%%%%%%%%%%%%
We complement the principle-based analysis by showing how different power indices can yield distinct, meaningful explanations even when they appear closely related on paper, using simple classifiers as examples.
The goal is not to ``pick a winner'', but to build intuition for what kinds of questions need to be answered when choosing an index.
In particular, we want to understand which coalition contexts an index implicitly treats as plausible, what it rewards, and how these choices interact with the structure of the underlying model.

We focus on a few recurring phenomena that are easy to overlook when only considering axioms or closed-form definitions.
First, indices that share many classical properties can still disagree on induced rankings, even in small examples, which matters whenever explanations are used to prioritise features, allocate responsibility, or justify interventions.
Second, when there is prior structure in how coalitions form, for example, through feature groups, dependencies, or natural evidence channels, indices that explicitly model such structure can yield qualitatively different attributions than indices that treat all coalitions as equally plausible, or than simply collapsing a group into a single unified player.
The implementation of our examples is openly available on GitHub at: \url{https://github.com/filipnaudot/powerXAI}.

%%%%%%%%%%%%%
\subsection{Ranking Discrepancies}
\label{section:ranking-example}
%%%%%%%%%%%%%
\begin{figure}[ht!]
  \centering
  \includegraphics[width=0.5\linewidth]{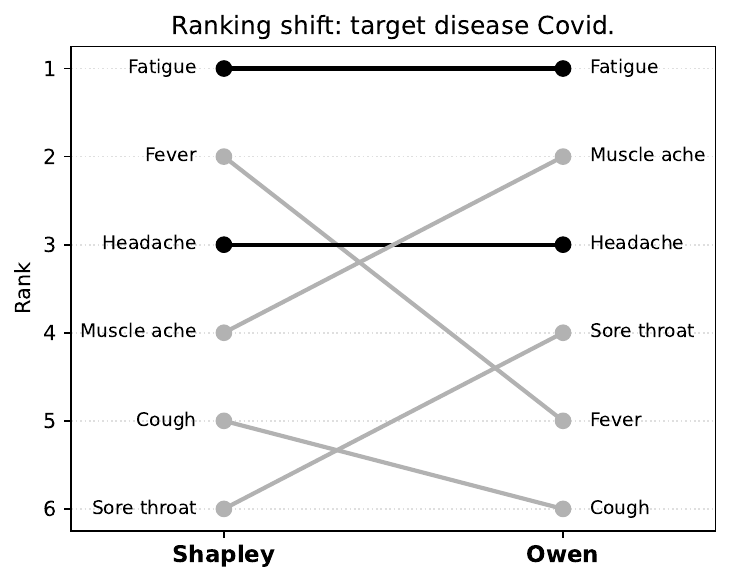}
  \caption{Ranking shift between the Shapley value and the Owen value on a disease prediction example.}
  \label{fig:shapley-banzhaf-ranking}
\end{figure}

Different power indices applied to the same model may induce different rankings of feature importance \cite{surprising-properties2001Saari}.
To illustrate this, consider a small disease-classification model with symptom features such as \emph{fever}, \emph{cough}, \emph{headache}, \emph{muscle ache}, \emph{sore throat}, and \emph{fatigue}.
For a fixed target class, we define the value of a feature coalition as the model's predicted probability for that class when exactly those features are present (cf.~\cite{lundberg2017SHAP}).
This naturally turns the model output into a cooperative game.

This example contains both individual effects and feature interactions.
Some symptoms raise the target prediction on their own, while others become important mainly through combinations with related symptoms.
Already at this level, it becomes clear that ``importance'' can be interpreted in different ways.
An index that emphasises how often a feature is decisive across coalition contexts may rank one symptom highly, while another index that aggregates contributions differently may place more weight on a symptom whose effect is spread across several interactions.
As a result, two indices can agree that the same small set of symptoms matters overall, yet still disagree on their relative ordering.

This illustrates that the choice of index is not merely a technical detail.
Even when two indices are conceptually similar, they can emphasise different aspects, thereby affecting the features' ranking.
\Cref{fig:shapley-banzhaf-ranking} illustrates this phenomenon for the disease example, where the induced rankings from the Shapley and Owen values differ.
When explanations are used to prioritise features, compare interventions, or justify decisions, these differences matter.
To judge whether a ranking is meaningful, one therefore has to look not only at the model, but also at the principles built into the chosen index and at whether those principles match the structure of the problem.

%%%%%%%%%%%%%
\subsection{When Group Structure Changes Attribution}
\label{section:coalition-example}
%%%%%%%%%%%%%
In many machine-learning settings, features do not act as isolated pieces of information, that is, they can be correlated.
Some features belong naturally together, either because they come from the same source or because they describe the same underlying phenomenon.
In such cases, treating all coalitions as equally likely can obscure influential features rather than reveal them.
The Owen value addresses this by incorporating an a priori coalition structure into the attribution process, and the difference in the resulting attributions can be substantial.
Crucially, as we will see, this is not the same as simply grouping features together under a set-based index.

To illustrate this, consider a fraud-detection setting where a classifier uses signals such as \emph{transaction amount}, \emph{odd hour}, \emph{location mismatch}, \emph{foreign IP}, and \emph{VPN usage} to predict different fraud types.
We create a simple classifier and use synthetic data (available in the GitHub repository).
For the \emph{Account Takeover} class, some of these signals naturally form a block: location, foreign IP, and VPN usage all describe the login environment and together point to the same underlying concern, namely that the account is accessed from a suspicious context.
This preference structure induces an a priori partition in which these features form one group, while the remaining fraud signals form other parts of the game.
We model the explanation problem as a cooperative game, where the value of a coalition reflects the model's output for the target class.

If we ignore this structure, a standard single-player attribution treats all features as separate contributors.
Once we define the grouping information, however, the Owen value tells a more nuanced story: attribution is first distributed across groups, and then split within each group.
In this way, the explanation reflects not only whether a feature is useful on its own, but also that some features belong to the same evidence channel, as illustrated in \Cref{fig:shapley-vs-owen}.
\begin{figure}[h]
    \centering
    \begin{subfigure}[t]{0.49\linewidth}
        \includegraphics[width=\linewidth]{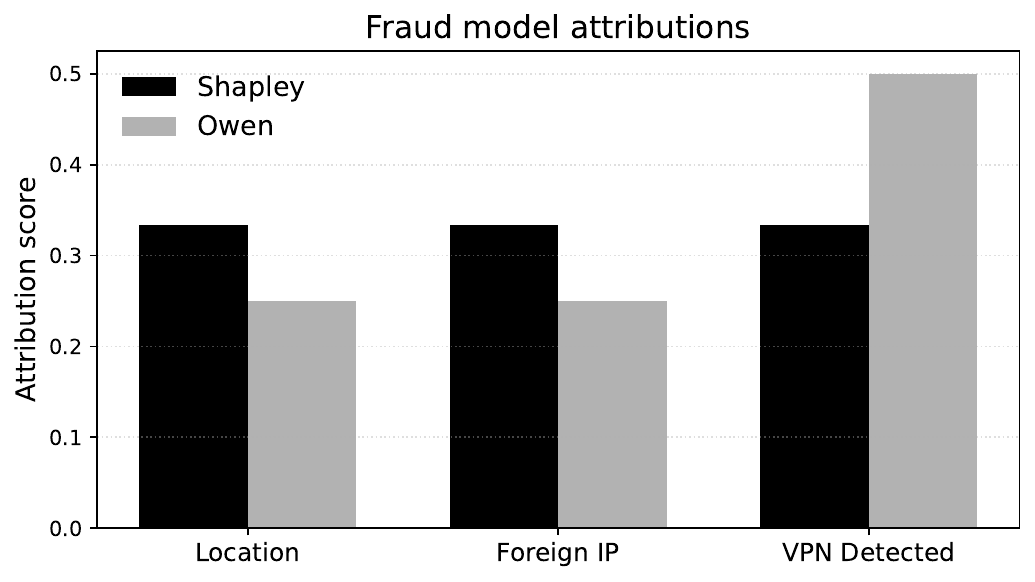}
        \caption{Standard versus group-aware attributions for the fraud example.
        Once the natural grouping of related signals is taken into account, the allocation of credit changes.}
        \label{fig:shapley-vs-owen}
    \end{subfigure}
    \hfill
    \begin{subfigure}[t]{0.49\linewidth}
        \includegraphics[width=\linewidth]{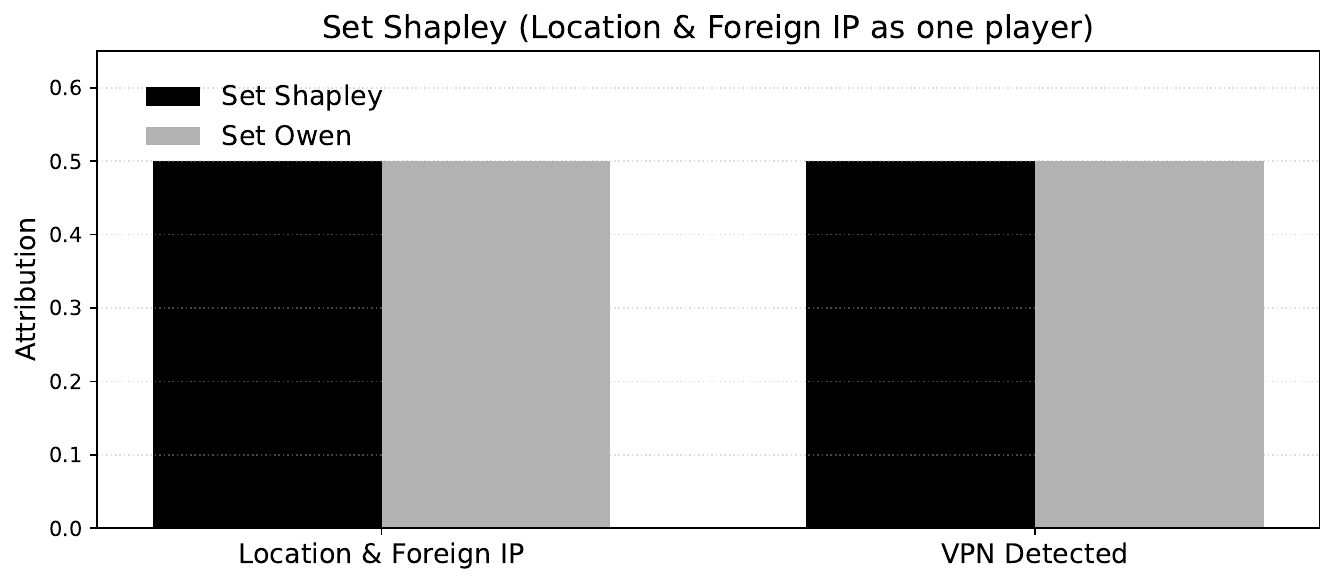}
        \caption{Set-based attributions when the grouped fraud signals are treated as a single composite player.
        The group's contribution is retained, but attribution to individual players within the group is lost.}
        \label{fig:setshapley-vs-setowen}
    \end{subfigure}
    \caption{Comparison of attribution schemes for the fraud-detection example.}
    \label{fig:comparison}
\end{figure}

One might wonder whether the same effect could be achieved more simply by treating the grouped signals as a single composite player and applying a standard set-based value.
However, doing so answers a different question.
A set-based index can tell us how much the grouped evidence contributes as a whole, but it does not attempt to distribute credit among the individual features inside that group.
The Owen value, by contrast, respects both levels simultaneously: it accounts for the structure between groups while still attributing credit to individual features within them.
When one instead turns to a set-based formulation and interprets the suspicious login signals as a single player, the internal dynamics are ignored, as illustrated in \Cref{fig:setshapley-vs-setowen}.

This difference in attribution has a direct practical implication:
If the goal is to understand which individual signals the model is most sensitive to, a standard single-player index may be appropriate.
If the goal is to reflect that several signals belong to the same coherent evidence channel, then an Owen-style attribution is more appropriate.
Alternatively, if we only care about the contribution of the grouped channel as a whole, then a set-based index may be sufficient.
Different power indices do not merely produce different numbers; they answer related, but fundamentally different questions.

To make this even clearer, one can compute the $\upsiempty$-value and ask how large a coalition must be for the model to reach a strong prediction for the target class.
In this example, the $\upsiempty$-value highlights that the model's strongest evidence emerges only when enough of the relevant fraud signals are present together, as illustrated in \Cref{fig:upsilon-vs-shapley-and-owen}.

\begin{figure}[ht!]
    \centering
    \includegraphics[width=0.6\linewidth]{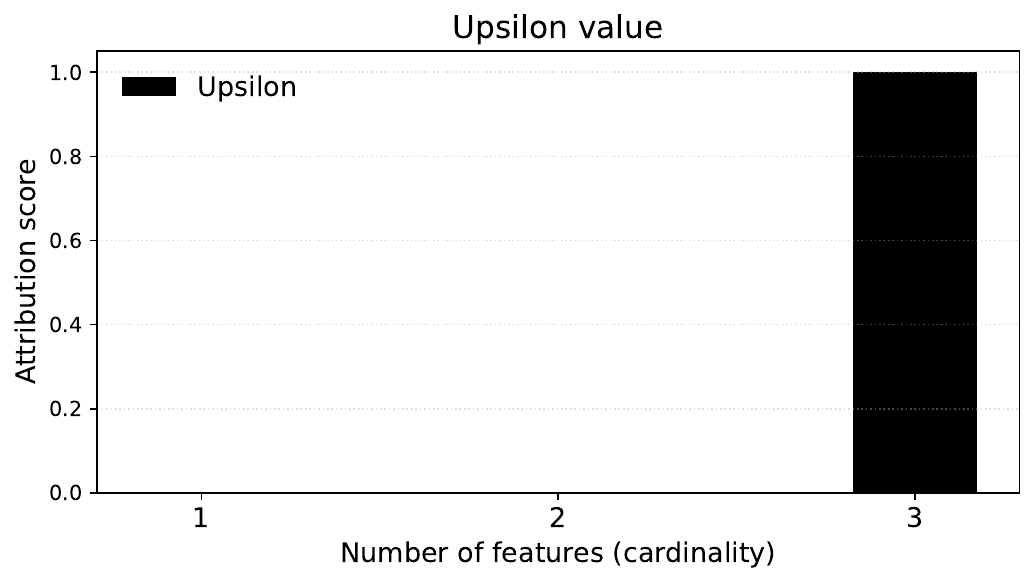}
    \caption{$\upsiempty$-value attribution for the fraud example.
        The $\upsiempty$-value highlights the coalition size at which the model's prediction becomes most valuable.}
    \label{fig:upsilon-vs-shapley-and-owen}
\end{figure}

%%%%%%%%%%%%%%%%%%%%%%%%%%%%%%%%
\section{Discussion}
\label{section:discussion}
%%%%%%%%%%%%%%%%%%%%%%%%%%%%%%%%
This paper set out to make the topic of power indices easier to navigate. 
There are many power indices with different mathematical properties, and it is not always obvious which index to use given a specific explainability scenario.
Rather than declaring a winner, we hope our examples and analyses provide researchers and practitioners with a clearer map of the terrain.
A useful place to start is with a question that precedes any choice of index: \emph{what is a player in our problem?}
This might seem like a technical detail, but it shapes everything that follows. 
If what matters is to understand the contribution of individual features, for instance, to rank them by importance or to decide which ones to keep, then single-player or set-based indices are the natural fit. 
These two power-index dimensions have the benefit that the same axioms hold, since grouping features can be seen as a simple relabelling.
In practice, the choice should be guided by what the explanation is meant to convey.
Use single-player indices when the units of interest are meaningful on their own, and set-based indices when it is only meaningful to work with coherent bundles of features that would naturally be interpreted or removed together.
Thus, feature groups should be treated as sets only when the grouping reflects a real semantic or operational unit, not merely because several quantities can technically be encoded into one player.
If, instead, the focus lies less on which features matter and more on how many tend to matter, for example, whether a model's predictions are typically driven by a few dominant features versus all features are similarly predictive, then the cardinality-based perspective may be more informative.
The $\upsiempty$-value, for instance, does not ask ``how important is feature X?'' but rather ``how much in general do coalitions improve when they grow?'' 
These are fundamentally different questions, and it is worth being deliberate about which one we are actually asking.

Once the player dimension is set, the next consideration is what properties matter in the specific use case.
The tables in \Cref{section:principle-analysis} are intended to support this decision.
A natural starting point is counterfactuality, the idea that a feature's attribution score should agree with what happens when it is removed from the full set of features. 
If removing a feature changes the outcome in one direction, the power index score should reflect this.
Only the counterfactual index guarantees this. 
It is worth noting, however, that what the value function actually measures is not always obvious in practice. 
In machine learning classifiers and neural networks, it typically captures change rather than quality, so the scores are relative quantities and should be interpreted with that in mind.
The Shapley and Banzhaf indices can, in principle, assign a score that does not align with the leave-one-out effect, because they average contributions across all coalition contexts rather than focusing on that single scenario.
Whether this is a problem depends on whether the attributions should be used to make direct decisions about individual features.
If so, then the counterfactual may be the appropriate index to choose.
If, however, the goal is a broader picture of how features contribute across all possible subsets, the averaging behaviour of the Shapley or Banzhaf indices may be more informative. 
Counterfactuality and several other desirable properties are not simultaneously satisfied by any single index. 
A practitioner has to decide which properties are most important in their context---a choice that often depends on factors such as how the results will be used and what kind of analysis the application requires.

Even when two indices agree on the overall structure of a problem, they can disagree on the relative ordering of players.
Since feature importance analyses are often used to rank features, prioritise them for further investigation, remove the least important ones, or explain a decision, this disagreement is not purely theoretical. 
Before committing to an index, it is worth asking: \emph{what would change if we used a different power index?} 
If the top-ranked feature shifts depending on the choice of index, this is a signal worth taking seriously, not necessarily a reason to distrust the indices, but a reason to understand what each one is actually measuring. 

The comparison in \Cref{section:examples} between the set-based Owen value and the set-based Shapley value highlights a more nuanced decision.
When features naturally group together, perhaps because they come from the same data source, encode the same concept, or are known to be strongly related, there are at least two ways to handle this. 
One option is to treat the group as a single composite player and apply a set-based index; this captures how much the group as a whole contributes, but says nothing specifically about the individuals within it.
Another option is to use the Owen value, which respects the group structure while still attributing credit to individual players.
The two approaches answer different questions (see \Cref{section:coalition-example}), and collapsing a group of features into a single player, while sometimes convenient, obscures the finer-grained distinctions within the group, a trade-off that may or may not be acceptable depending on the goal.
Finally, the move to cardinality-based indices deserves a comment.
As we showed in \Cref{section:cardinality-power-indices}, moving to cardinality-based indices removes distinctions that depend on player identity, but there are still meaningful differences between cardinality indices.
In particular, the cardinality-based Shapley value coincides with the $\upsiempty$-value, whereas the cardinality-based Banzhaf value $\cardbanzhafempty$ remains distinct because it applies a cardinality-dependent weighting to the same adjacent-cardinality changes.
In one sense, this is reassuring since the cardinality setting is coarser with respect to player identities.
In another sense, it still leaves meaningful choices about how changes across coalition sizes should be weighted.
By abstracting away player identity, it gives up the ability to distinguish between coalitions that happen to have the same size but different compositions.
Whether that abstraction is appropriate, a strength, or a weakness depends, again, on the question being asked.

Recent critiques of Shapley value-based explainability show that some XAI applications of Shapley values can conflict with formal notions of feature relevance~\cite{explainability_is_not_a_game}.
However, subsequent work \cite{updated_exp_is_NOT_a_game} clarifies that these failures need not be attributed to the Shapley value itself, but to the characteristic function used to define the underlying cooperative game.
This distinction is central to our perspective, which is that power indices provide suitable explanations only when the underlying game is defined to match the intended explanatory question.
That is, only after one specifies the players, the value function, the attribution dimension, and the principles one expects the resulting scores to satisfy.

Taken together, the analysis in this paper suggests that the choice of power index is best understood as a decision driven by the structure of the explanation problem and the characteristics that are important in the problem's context. 
The principles introduced in \Cref{section:principles} are intended to make that choice more explicit and better informed.
Accordingly, the scope of this paper is conceptual and axiomatic rather than algorithmic.
We do not study the computational complexity of the indices extensively, nor approximation and heuristic methods for computing them more efficiently.
We also do not aim to survey full application pipelines or explainability frameworks beyond power-index-based approaches.
Our contribution is to clarify which kinds of attribution questions different power indices answer, and which formal principles they satisfy.

%%%%%%%%%%%%%%%%%%%%%%%%%%%%%%%%
\section{Conclusion}
\label{section:conclusion}
%%%%%%%%%%%%%%%%%%%%%%%%%%%%%%%%
Power indices offer a principled way to formalise attribution in XAI, but their interpretation depends on the dimension at which power is assigned, that is, single players, sets of players, or cardinalities.
This distinction is intended as a step toward a more explicit use of power indices in explainability and highlights that different power indices do not merely provide alternative numerical scores but often answer different attribution questions.
We have reviewed representative power indices for these dimensions, generalised existing indices, and analysed them through a set of well-established principles.
Our analysis highlights, also supported by intuitive examples, that there is no single best index.
Instead, the choice of power index depends on the explanation task.
The cardinality-based setting, inspired by~\cite{Upsilon-values2024}, also shows that removing player identity changes the nature of the attribution problem. Rather than asking which players or sets matter, cardinality-based indices ask how the game value changes as coalition size increases.
Our guide can be used before choosing and applying a power index, as a way to rule out indices that do not match the problem at hand.
This makes the choice of power index an explicit part of the problem setup, rather than an often arbitrary design decision in an explanation generation pipeline.
The guide can also help users make sense of cases where indices lead to varying outcomes, since the indices may be answering different questions.

Future work may extend our analysis in several directions.
One such direction is by considering aspects pertaining to computational costs, and another is to examine concrete explanation pipelines.

%%%%%%%%%%%%%%%%%%%%%%%%%%%%%%%%
\subsubsection*{Acknowledgments}
\label{section:acknowledgments}
%%%%%%%%%%%%%%%%%%%%%%%%%%%%%%%%
This work was partially supported by the Wallenberg AI, Autonomous Systems and Software Program (WASP) funded by the Knut and Alice Wallenberg Foundation.

% \bibliography{ref}
% \bibliographystyle{unsrtnat}

\printcredits

\bibliographystyle{cas-model2-names}

\bibliography{ref}

\begin{thebibliography}{18}
\expandafter\ifx\csname natexlab\endcsname\relax\def\natexlab#1{#1}\fi
\providecommand{\url}[1]{\texttt{#1}}
\providecommand{\href}[2]{#2}
\providecommand{\path}[1]{#1}
\providecommand{\DOIprefix}{doi:}
\providecommand{\ArXivprefix}{arXiv:}
\providecommand{\URLprefix}{URL: }
\providecommand{\Pubmedprefix}{pmid:}
\providecommand{\doi}[1]{\href{http://dx.doi.org/#1}{\path{#1}}}
\providecommand{\Pubmed}[1]{\href{pmid:#1}{\path{#1}}}
\providecommand{\bibinfo}[2]{#2}
\ifx\xfnm\relax \def\xfnm[#1]{\unskip,\space#1}\fi
%Type = Inproceedings
\bibitem[{Amgoud and Ben-Naim(2022)}]{Leila2022-axiomatic-foundations}
\bibinfo{author}{Amgoud, L.}, \bibinfo{author}{Ben-Naim, J.}, \bibinfo{year}{2022}.
\newblock \bibinfo{title}{Axiomatic foundations of explainability}, in: \bibinfo{editor}{Raedt, L.D.} (Ed.), \bibinfo{booktitle}{Proceedings of the Thirty-First International Joint Conference on Artificial Intelligence, {IJCAI-22}}, \bibinfo{publisher}{International Joint Conferences on Artificial Intelligence Organization}. pp. \bibinfo{pages}{636--642}.
\newblock \URLprefix \url{https://doi.org/10.24963/ijcai.2022/90}, \DOIprefix\doi{10.24963/ijcai.2022/90}. \bibinfo{note}{main Track}.
%Type = Article
\bibitem[{Banzhaf(1965)}]{Banzhaf1965}
\bibinfo{author}{Banzhaf, J.}, \bibinfo{year}{1965}.
\newblock \bibinfo{title}{Weighted voting doesn't work: {A} mathematical analysis}.
\newblock \bibinfo{journal}{Rutgers Law Review} \bibinfo{volume}{19}, \bibinfo{pages}{317--343}.
%Type = Article
\bibitem[{Chen et~al.(2023)Chen, Covert, Lundberg and Lee}]{Chen2023}
\bibinfo{author}{Chen, H.}, \bibinfo{author}{Covert, I.C.}, \bibinfo{author}{Lundberg, S.M.}, \bibinfo{author}{Lee, S.I.}, \bibinfo{year}{2023}.
\newblock \bibinfo{title}{Algorithms to estimate shapley value feature attributions}.
\newblock \bibinfo{journal}{Nature Machine Intelligence} \bibinfo{volume}{5}, \bibinfo{pages}{590--601}.
\newblock \DOIprefix\doi{10.1038/s42256-023-00657-x}.
%Type = Article
\bibitem[{Dubey and Shapley(1979)}]{mathematical_properties_Banzhaf_value}
\bibinfo{author}{Dubey, P.}, \bibinfo{author}{Shapley, L.S.}, \bibinfo{year}{1979}.
\newblock \bibinfo{title}{Mathematical properties of the banzhaf power index}.
\newblock \bibinfo{journal}{Mathematics of Operations Research} \bibinfo{volume}{4}, \bibinfo{pages}{99--131}.
\newblock \URLprefix \url{http://www.jstor.org/stable/3689345}.
%Type = Article
\bibitem[{Grabisch and Roubens(1999)}]{Grabisch1999}
\bibinfo{author}{Grabisch, M.}, \bibinfo{author}{Roubens, M.}, \bibinfo{year}{1999}.
\newblock \bibinfo{title}{An axiomatic approach to the concept of interaction among players in cooperative games}.
\newblock \bibinfo{journal}{International Journal of Game Theory} \bibinfo{volume}{28}, \bibinfo{pages}{547--565}.
\newblock \DOIprefix\doi{10.1007/s001820050125}.
%Type = Article
\bibitem[{Kampik et~al.(2024)Kampik, Potyka, Yin, Cyras and Toni}]{KampikPYCT24-principle}
\bibinfo{author}{Kampik, T.}, \bibinfo{author}{Potyka, N.}, \bibinfo{author}{Yin, X.}, \bibinfo{author}{Cyras, K.}, \bibinfo{author}{Toni, F.}, \bibinfo{year}{2024}.
\newblock \bibinfo{title}{Contribution functions for quantitative bipolar argumentation graphs: {A} principle-based analysis}.
\newblock \bibinfo{journal}{Int. J. Approx. Reason.} \bibinfo{volume}{173}, \bibinfo{pages}{109255}.
\newblock \URLprefix \url{https://doi.org/10.1016/j.ijar.2024.109255}, \DOIprefix\doi{10.1016/J.IJAR.2024.109255}.
%Type = Inproceedings
\bibitem[{Lundberg and Lee(2017)}]{lundberg2017SHAP}
\bibinfo{author}{Lundberg, S.M.}, \bibinfo{author}{Lee, S.I.}, \bibinfo{year}{2017}.
\newblock \bibinfo{title}{A unified approach to interpreting model predictions}, in: \bibinfo{booktitle}{Proceedings of the 31st International Conference on Neural Information Processing Systems}, \bibinfo{publisher}{Curran Associates Inc.}, \bibinfo{address}{Red Hook, NY, USA}. p. \bibinfo{pages}{4768–4777}.
%Type = Article
\bibitem[{Marques-Silva and Huang(2024)}]{explainability_is_not_a_game}
\bibinfo{author}{Marques-Silva, J.}, \bibinfo{author}{Huang, X.}, \bibinfo{year}{2024}.
\newblock \bibinfo{title}{Explainability is not a game}.
\newblock \bibinfo{journal}{Commun. ACM} \bibinfo{volume}{67}, \bibinfo{pages}{66–75}.
\newblock \URLprefix \url{https://doi.org/10.1145/3635301}, \DOIprefix\doi{10.1145/3635301}.
%Type = Misc
\bibitem[{Marques-Silva et~al.(2025)Marques-Silva, Huang and Letoffe}]{updated_exp_is_NOT_a_game}
\bibinfo{author}{Marques-Silva, J.}, \bibinfo{author}{Huang, X.}, \bibinfo{author}{Letoffe, O.}, \bibinfo{year}{2025}.
\newblock \bibinfo{title}{The explanation game -- rekindled (extended version)}.
\newblock \URLprefix \url{https://arxiv.org/abs/2501.11429}, \href{http://arxiv.org/abs/2501.11429}{\tt arXiv:2501.11429}.
%Type = Masterthesis
\bibitem[{Naudot(2025)}]{naudot2025scalable}
\bibinfo{author}{Naudot, F.}, \bibinfo{year}{2025}.
\newblock \bibinfo{title}{Scalable Feature Attribution in LLM-Based Recommender Systems}.
\newblock \bibinfo{type}{Degree project in computing science and engineering}. Umea University.
\newblock \URLprefix \url{https://urn.kb.se/resolve?urn=urn:nbn:se:umu:diva-239871}. \bibinfo{note}{diVA, id: diva2:1966000}.
%Type = Misc
\bibitem[{Naudot et~al.(2025)Naudot, Sundqvist and Kampik}]{naudot2025llmshap}
\bibinfo{author}{Naudot, F.}, \bibinfo{author}{Sundqvist, T.}, \bibinfo{author}{Kampik, T.}, \bibinfo{year}{2025}.
\newblock \bibinfo{title}{llmshap: A principled approach to llm explainability}.
\newblock \URLprefix \url{https://arxiv.org/abs/2511.01311}, \href{http://arxiv.org/abs/2511.01311}{\tt arXiv:2511.01311}.
%Type = Inproceedings
\bibitem[{Owen(1977)}]{owen1977}
\bibinfo{author}{Owen, G.}, \bibinfo{year}{1977}.
\newblock \bibinfo{title}{Values of games with a priori unions}, in: \bibinfo{editor}{Henn, R.}, \bibinfo{editor}{Moeschlin, O.} (Eds.), \bibinfo{booktitle}{Mathematical Economics and Game Theory}, \bibinfo{publisher}{Springer Berlin Heidelberg}, \bibinfo{address}{Berlin, Heidelberg}. pp. \bibinfo{pages}{76--88}.
%Type = Inproceedings
\bibitem[{Rozemberczki et~al.(2022)Rozemberczki, Watson, Bayer, Yang, Kiss, Nilsson and Sarkar}]{ijcai2022p778}
\bibinfo{author}{Rozemberczki, B.}, \bibinfo{author}{Watson, L.}, \bibinfo{author}{Bayer, P.}, \bibinfo{author}{Yang, H.T.}, \bibinfo{author}{Kiss, O.}, \bibinfo{author}{Nilsson, S.}, \bibinfo{author}{Sarkar, R.}, \bibinfo{year}{2022}.
\newblock \bibinfo{title}{The shapley value in machine learning}, in: \bibinfo{editor}{Raedt, L.D.} (Ed.), \bibinfo{booktitle}{Proceedings of the Thirty-First International Joint Conference on Artificial Intelligence, {IJCAI-22}}, \bibinfo{publisher}{International Joint Conferences on Artificial Intelligence Organization}. pp. \bibinfo{pages}{5572--5579}.
\newblock \DOIprefix\doi{10.24963/ijcai.2022/778}. \bibinfo{note}{survey Track}.
%Type = Article
\bibitem[{Saari and Sieberg(2001)}]{surprising-properties2001Saari}
\bibinfo{author}{Saari, D.G.}, \bibinfo{author}{Sieberg, K.K.}, \bibinfo{year}{2001}.
\newblock \bibinfo{title}{Some surprising properties of power indices}.
\newblock \bibinfo{journal}{Games and Economic Behavior} \bibinfo{volume}{36}, \bibinfo{pages}{241--263}.
\newblock \URLprefix \url{https://www.sciencedirect.com/science/article/pii/S0899825600908194}, \DOIprefix\doi{https://doi.org/10.1006/game.2000.0819}.
%Type = Incollection
\bibitem[{Shapley(1953)}]{shapley1953}
\bibinfo{author}{Shapley, L.S.}, \bibinfo{year}{1953}.
\newblock \bibinfo{title}{A value for n-person games}, in: \bibinfo{editor}{Kuhn, H.W.}, \bibinfo{editor}{Tucker, A.W.} (Eds.), \bibinfo{booktitle}{Contributions to the Theory of Games II}. \bibinfo{publisher}{Princeton University Press}, \bibinfo{address}{Princeton}, pp. \bibinfo{pages}{307--317}.
%Type = Article
\bibitem[{Shapley and Shubik(1954)}]{shapley-shubik1954}
\bibinfo{author}{Shapley, L.S.}, \bibinfo{author}{Shubik, M.}, \bibinfo{year}{1954}.
\newblock \bibinfo{title}{A method for evaluating the distribution of power in a committee system}.
\newblock \bibinfo{journal}{The American Political Science Review} \bibinfo{volume}{48}, \bibinfo{pages}{787--792}.
\newblock \URLprefix \url{http://www.jstor.org/stable/1951053}.
%Type = Inproceedings
\bibitem[{Sundararajan and Najmi(2020)}]{pmlr-v119-sundararajan20b}
\bibinfo{author}{Sundararajan, M.}, \bibinfo{author}{Najmi, A.}, \bibinfo{year}{2020}.
\newblock \bibinfo{title}{The many shapley values for model explanation}, in: \bibinfo{editor}{III, H.D.}, \bibinfo{editor}{Singh, A.} (Eds.), \bibinfo{booktitle}{Proceedings of the 37th International Conference on Machine Learning}, \bibinfo{publisher}{PMLR}. pp. \bibinfo{pages}{9269--9278}.
%Type = Article
\bibitem[{Torra(2024)}]{Upsilon-values2024}
\bibinfo{author}{Torra, V.}, \bibinfo{year}{2024}.
\newblock \bibinfo{title}{{$\Upsilon$-Values: Power Indices à La Orness for Nonadditive Measures}}.
\newblock \bibinfo{journal}{IEEE Transactions on Fuzzy Systems} \bibinfo{volume}{32}, \bibinfo{pages}{4099--4108}.
\newblock \DOIprefix\doi{10.1109/TFUZZ.2024.3392268}.

\end{thebibliography}



 \clearpage
%%%%%%%%%%%%%%%%%%%%%%%%%%%%%%%%%%%%%%%%%%%%%%%%%%%%
\appendix
\section*{Appendix}
\label{section:appendix}
%%%%%%%%%%%%%%%%%%%%%%%%%%%%%%%%%%%%%%%%%%%%%%%%%%%%

%%%%%%%%%%%%%%%%%%%%%%%%%%%%%%%%%%%%%%%%%%%
% BEGINNING PRINCIPLE-BASED ANALYSIS PROOFS
%%%%%%%%%%%%%%%%%%%%%%%%%%%%%%%%%%%%%%%%%%%
\section{Principle-based Analysis Proofs}\label{appendix:principle-based-analysis}
This part of the appendix contains the explicit proofs of the results stated in Section~\ref{section:principle-analysis}.
We restate the respective theorems and propositions before their proofs. 
\quanimpliesimple*
\begin{proof}
    Suppose a power index $\powerindexempty$ satisfies quantitative counterfactuality.
    As per the definition,
    \begin{equation}
    \label{eq:counterfactuality-condition}
        \mu_i(v)=v(N)-v(N\setminus\{i\})
    \end{equation} for any game $(N,v)$. 
    Now, considering three distinct cases based on the positivity of the counterfactual power index $\counterfact{i}{v}$ yields:
    \begin{itemize}
        \item 
        If $\mu_i(v)>0$, then by \Cref{eq:counterfactuality-condition} $v(N)-v(N\setminus\{i\})>0$, i.e. 
        $v(N)>v(N\setminus\{i\})$,
        \item 
        If $\mu_i(v)=0$, then by \Cref{eq:counterfactuality-condition} $v(N)-v(N\setminus\{i\})=0$, i.e. 
        $v(N)=v(N\setminus\{i\})$,
        \item 
        If $\mu_i(v)<0$, then by \Cref{eq:counterfactuality-condition} $v(N)-v(N\setminus\{i\})<0$, i.e. 
        $v(N)<v(N\setminus\{i\})$.
    \end{itemize}
    The three conditions above are precisely the ones required for counterfactuality to hold.
    This closes the proof. 
\end{proof}
\counterfactualnonstandardpositive*
\begin{proof}
    We prove that the counterfactual power index satisfies quantitative counterfactuality, then, by \Cref{proposition:quan-implies-simple} counterfactuality follows.
    For the proof of the quantitative part, notice that the definition of counterfactual power index $\counterfact{i}{v} = v(N) - v(N \setminus \{i\})$ coincides with that of quantitative counterfactuality in Principle~\ref{principle:quantitative-counterfactuality}. 
\end{proof}
\counterfactualpositive*
\begin{proof}
    For proving that counterfactuality satisfies symmetry, assume that players $i,j$ satisfy $v(S\cup \{i\})=v(S\cup \{j\})$ for every $S\subseteq N\setminus\{i,j\}$, for some arbitrary game $(N,v)$.
    The assumption implies: 
    \begin{equation}
    \label{eq:symmetry-condition}
        v(N\setminus\{i\})=v(N\setminus\{j\}).
    \end{equation}
    Now, we need to show that $\counterfact{i}{v} =\counterfact{j}{v}$; towards that goal,
    expand the definition of $\counterfact{i}{v}$, which yields: 
    \begin{align*}
        \counterfact{i}{v}
        &= v(N) - v(N \setminus\{i\})\\
        &= v(N) - v(N \setminus\{j\}) = \counterfact{j}{v}, 
    \end{align*}
    where the second equality follows from \Cref{eq:symmetry-condition}.
    This closes the proof for symmetry.

    For the case of a null player, suppose that there is a null player $i$ in the game $(N,v)$, that is,
    $v(S \cup \{i\}) = v(S)$ for every $S \subseteq N \setminus \{i\}$.
    Now, this assumption implies $v(N) = v(N \setminus \{i\})$, or equivalently, 
    \begin{equation}
    \label{eq:counterfactuality-null-player-condition}
        v(N) - v(N \setminus \{i\}) = 0.
    \end{equation}
    Since $\counterfact{i}{v} = v(N) - v(N \setminus \{i\})$, \Cref{eq:counterfactuality-null-player-condition} implies $\counterfact{i}{v} = 0$. 
    This closes the proof for the null player.

    Finally, we consider the case for weak anonymity. 
    Consider any permutation $\pi$ over the set of players $N$, for the game $(N,v)$. It is immediate that:
    \begin{equation}
    \label{eq:permutation-condition}
        \hat{\pi}(N)=N\hspace{1 em}\text{ and }\hspace{1 em}\hat{\pi}(N\setminus\{i\})=N\setminus\{\pi(i)\}.
    \end{equation}
    Therefore, for any player $i\in N$ we have the following chain of equality from $\counterfact{\pi(i)}{v_\pi}$ to $\counterfact{i}{v}$. 
    \begin{align*}
        \counterfact{\pi(i)}{v_\pi}
        &= v_\pi(\hat\pi(N))-v_\pi(\hat\pi(N)\setminus\{\pi(i)\})
        \\
        &= v(N)-v_\pi(\hat\pi(N\setminus \{i\}))
        \\
        &= v(N) - v(N\setminus\{i\}) = \counterfact{i}{v},
    \end{align*}
    where the second equality is by \Cref{eq:permutation-condition}. This closes the proof. 
\end{proof}
\simplestcaseowen*
\begin{proof}
    Let $(N,v)$ be a game with $N = \{1,\dots,m\}$ and with only singleton blocks $M_N = \{S^1, \dots, S^m\}$, i.e. $|S^j| = 1$ for every $1 \leq j \leq m$. 
    WLOG assume that $S^j = \{j\}$ and remember the general definition of Owen values:
    \begin{equation}
    \label{eq:owen-values-case-study}
        \owen{i}{v} = \sum_{T \subseteq M_N \setminus \{S^k\}} \sum_{K \subseteq S^k \setminus \{i\}} \frac{|K|! \left(|S^k| - |K| - 1\right)! |T|! \left(|M_N| - |T| - 1\right)!}{|S^k|! |M_N|!} \left( v\left(Q \cup K \cup \left\{i \right\}\right) - v\left(Q \cup K \right) \right).
    \end{equation}
    Now, we refine Equation~\ref{eq:owen-values-case-study} as per our case.
    Since each $S^j$ is singleton, any $K \subseteq S^k \setminus \{i\}$ means $K = \emptyset$.
    Hence, there is only one such $K$ with $|K|=0$.
    Substituting the above in Equation~\ref{eq:owen-values-case-study} yields:
    \begin{align*}
        \owen{i}{v} &=
        \sum_{T \subseteq M_N \setminus \{S^k\}} \frac{0! \left(1 - 0 - 1\right)! |T|! \left(m - |T| - 1\right)!}{1! \, m!} \left( v\left(Q \cup \left\{i\right\}\right) - v\left(Q \right) \right)
        \\
        &=
        \sum_{T \subseteq M_N \setminus \{S^k\}} \frac{|T|! \left(m - |T| - 1\right)!}{m!} \left( v\left(Q \cup \left\{i\right\}\right) - v\left(Q \right) \right)
        \\
        &=
        \shap{i}{v},
    \end{align*}
    where the last equality holds because $T \mapsto Q$ is a bijection from the subsets of $M_N \setminus \{S^k\}$ onto the subsets of $N \setminus \{i\}$ with $|T| = |Q|$, so the sum coincides with \Cref{powerindex:shapley}.
\end{proof}

\owennegative*
\begin{proof}
    Pick any non-classical principle out of success, counterfactuality or quantitative counterfactuality, and assume for the sake of contradiction that Owen values satisfy it. 
    This implies that the chosen principle is also satisfied for singleton partitions.
    Now, we can apply \Cref{theorem:simplest-case-owen} to obtain that the chosen principle is satisfied by Shapley values. 
    This is a contradiction as  Shapley satisfied none of the non-classical principles.  
\end{proof}

\owenSymmetryNegativeProof*
\begin{proof}
    Let $N = \{x_1, x_2, x_3\}$ with a priori coalition structure $M_N = \{\{x_1, x_2\}, \{x_3\}\}$, and let
    \[
        v(S) :=
        \begin{cases}
            1, &\text{if $|S| \geq 2$};\\
            0, &\text{otherwise}.
        \end{cases}
    \]
    Since $v(S)$ depends only on $|S|$, for all $i, j \in N$ and every $S \subseteq N \setminus \{i, j\}$ we have $|S \cup \{i\}| = |S \cup \{j\}|$ and hence $v(S \cup \{i\}) = v(S \cup \{j\})$.
    Any two players are therefore symmetric.

    Player $x_3$ forms its own block, so the only $K \subseteq \{x_3\} \setminus \{x_3\}$ is $K = \emptyset$, and $T$ ranges over $\emptyset$ and $\{\{x_1,x_2\}\}$, giving $Q = \emptyset$ and $Q = \{x_1,x_2\}$ respectively.
    Both marginal contributions vanish,
    \[
        v(\{x_3\}) - v(\emptyset) = 0,
        \qquad
        v(\{x_1,x_2,x_3\}) - v(\{x_1,x_2\}) = 1 - 1 = 0,
    \]
    so $\owen{x_3}{v} = 0$.
    By \Cref{theorem:owen-classical-principles}, $\owenempty$ is efficient, so $\owen{x_1}{v} + \owen{x_2}{v} = v(N) - \owen{x_3}{v} = 1 - 0 = 1$, and therefore at least one of $\owen{x_1}{v}$ and  $\owen{x_2}{v}$ is non-zero.
    That player is symmetric with $x_3$ but receives a different payout, contradicting symmetry.
\end{proof}

\nonstanadardprinciplesofmonotonegames*
\begin{proof}
    \begin{enumerate}
        \item[i.]
        Let $(N,v)$ be an arbitrary monotone game on the player set $N=\{i_1,\dots,i_n\}$.
        Notice that, if $v(N)\leq v(\emptyset)$, then $v(N)=0$, and 
        $$
        0=v(\emptyset)\leq v(S)\leq v(N)=0,
        $$
        by monotonicity.
        The latter condition implies $v(S)=0$ for every coalition $S \subseteq N$, which leads to the trivial game.
        Hence, for the sake of non-triviality, we assume $v(N)>0$. 
        
        We claim that there is at least one $D\subseteq N$ and $i \in N$ such that $v(D\cup\{i\})-v(D)>0$. 
        For the sake of contradiction, assume it is not the case.
        Then for every $S$ and $j$ we have $v(S\cup\{j\})=v(S)$ by monotonicity.
        This leads to the chain of equality:
        $$
        v(N)=v(N\setminus\{i_1\})=v(N\setminus\{i_1,i_2\})=...=v(\emptyset)
        $$
        
        and that implies $v(N)=0$, contradicting our initial assumption. 
        Therefore, we have: 
        \begin{align*}
            \banzhaf{i}{v}
            &=
            \frac{1}{2^{n-1}}\sum_{S \subseteq N \setminus\{i\}}(v(S\cup \{i\})-v(S))\geq \frac{v(D\cup\{i\})-v(D)}{2^{n-1}}>0.
            \\
            \shap{i}{v}
            &=
            \sum_{S \subseteq N \setminus\{i\}}\frac{|S|!(n-|S|-1)!}{n!}(v(S\cup \{i\})-v(S))
            \\
            &\geq
            \frac{|D|!(n-|D|-1)!}{n!}{(v(D\cup\{i\})-v(D))}>0.
        \end{align*}
        This proves success for that class of monotonic games under the Banzhaf and Shapley power indices. 
        For Owen values, by a similar argument as above, there is a set of partitions $S^1, \dots ,S^i\in M_N=\{S^1, \dots, S^m\}$ s.t. $v(S^1\cup \dots \cup S^i)=0$, and $v(S^1\cup \dots \cup S^i \cup S^{i+1})>0$. 
        Repeating the same argument, there are subsets $K, \{j\}\subseteq S^{i+1}$ s.t.
        \begin{equation}
        \label{eq:owen-proof-ineq}
            v(S^1\cup \dots \cup S^i \cup K \cup \{j\})-v(S^1\cup \dots \cup S^i \cup K)>0.
        \end{equation} 
        For any $T \subseteq M_N \setminus \{S^{i+1}\}$, let $Q = \bigcup_{S^j \in T} S^j$.
        Now, following the equalities:
        \begin{align*}
            \owen{j}{v} &= 
            \sum_{T \subseteq M_N \setminus \{S^{i+1}\}} \sum_{U \subseteq S^{i+1} \setminus \{j\}} \scalebox{0.95}{$\frac{|U|! \left(|S^{i+1}| - |U| - 1\right)! |T|! \left(|M_N| - |T| - 1\right)!}{|S^{i+1}|! |M_N|!}$}
            \left( v\left(Q \cup U \cup \left\{j\right\}\right) - v\left(Q \cup U \right) \right)
            \\
            &\geq \scalebox{0.95}{$\frac{|K|! \left(|S^{i+1}| - |K| - 1\right)! |T|! \left(|M_N| - |T| - 1\right)!}{|S^{i+1}|! |M_N|!}$}
            v(S^1\cup \dots \cup S^i \cup K \cup \{j\})-v(S^1\cup \dots \cup S^i \cup K)
            \\
            &=\frac{1}{B} v(S^1\cup \dots \cup S^i \cup K \cup \{j\})-v(S^1\cup \dots \cup S^i \cup K),
        \end{align*}
        where $B$ is the normalising factor in the last inequality. 
        By \Cref{eq:owen-proof-ineq}, it follows:
        \begin{equation*}
            \owen{j}{v}\geq \frac{1}{B} v(S^1\cup \dots \cup S^i \cup K \cup \{j\})-v(S^1\cup \dots \cup S^i \cup K)>0,
        \end{equation*}
        which proves success for the class of monotonic games under Owen values.
        \item[ii.]
        Again, assume $(N,v)$ to be an arbitrary strictly monotone game. 
        Following the condition on monotonicity,  every $i\in N$ and $D\subseteq N\setminus\{i\}$ satisfies $v(D\cup\{i\})-v(D)>0$. 
        Hence, to verify that $\shapempty,\banzhafempty$ and $\owenempty$ satisfy counterfactuality, we only need to establish the first condition related to counterfactuality, that is, $v(N)-v(N\setminus \{i\})>0$.
        However, the latter condition again follows from the strict monotonicity condition. 
    \end{enumerate}
\end{proof}
\setbasedanalysis*
\begin{proof}
    Let $A$ be any set of players, and $\playerpartition$ be a set-based partition of it.
    Fix any game $(A,v)$ with $v: 2^A \rightarrow \mathbb{R}$, and consider the set-based power index $\setpowerindexempty$.
    Each set-based index is its single-player counterpart applied to the game $(\playerpartition, \setgame)$, whose players are the blocks of $\playerpartition$.
    Since every principle in \Cref{section:principles} quantifies over arbitrary games, a principle satisfied by $\powerindexempty$ on all games holds in particular on $(\playerpartition, \setgame)$, and hence is satisfied by $\setpowerindexempty$.

    For the negative direction, suppose a principle is not satisfied by 
    $\powerindexempty$. 
    The counterexample of that fact for $\powerindexempty$ can be translated to the counterexample vis-a-vis $\setpowerindexempty$ in the following way:
    The feature set $A = \{x_1, \ldots, x_n\}$ in the original example is now partitioned as $\playerpartition$ in the new example construction, that is, $\playerpartition = \{\{x_1\}, \ldots, \{x_n\}\}$. 
    Now, the new players are the singleton partition of $A$, that is, $\playerpartition$. 
    The function $\setgame$ corresponding to $\playerpartition$ is defined as:
    $$
    \setgame(S) = v\left(\bigcup_{s\in S} s\right).
    $$
    For instance, $\setgame(\{\{x_1\},\{x_2\}\})=v(\{x_1,x_2\})$.
    For the Owen value, the a priori coalition structure is translated in the same way so each block $S^k \in M_A$ becomes the block $\{\{x\} \mid x \in S^k\}$ of $M_\playerpartition$.
    Since $x \mapsto \{x\}$ is a bijection from $A$ onto $\playerpartition$ with $\setgame(\{\{x\} \mid x \in S\}) = v(S)$ for every $S \subseteq A$, the two games are identical up to renaming and the counterexample carries over.
\end{proof}
\cardqunatitativetonormal*
\begin{proof}
    Suppose $\cardpowerindexempty$ satisfies quantitative counterfactuality, and let $(N,v)$ be an arbitrary but fixed game.
    We will prove that $\cardpowerindexempty$ also satisfies counterfactuality. 
    Straight from the definition, it follows that:
    $$
    \cardpowerindex{i}{v}=v(N)-\left(\frac{1}{\binom{|N|}{i}}\sum_{|S|=i}v(S)\right),
    $$ 
    and hence whenever $\cardpowerindex{i}{v}>0$, then: 
    \begin{align*}
        &\hspace{2.7 em}\cardpowerindex{i}{v}>0
        \\
        &\iff v(N)-\left(\frac{1}{\binom{|N|}{i}}\sum_{|S|=i}v(S)\right)>0
        &\text{due to the equality above.}\\
        &\iff v(N)>\left(\frac{1}{\binom{|N|}{i}}\sum_{|S|=i}v(S)\right).
    \end{align*}
    A similar argument holds when $>$ is changed by $=,<$, in the cases of $\cardpowerindexempty=0$ and $\cardpowerindexempty<0$, respectively.
    This ends the proof.
\end{proof}

%%%%%%%%%%%%%%%%%%%%%%%%%%%%%%%%%%%%%%%%%%%%%
% CARDINALITY ANALYSIS
%%%%%%%%%%%%%%%%%%%%%%%%%%%%%%%%%%%%%%%%%%%%%
\cardBanzhafNegativeEfficiency*
\begin{proof}
    Consider the game over the feature set $N = \{x_1, x_2, x_3\}$ defined by
    \[
        v(S)=
        \begin{cases}
        1; &\text{if $S\neq \emptyset$},\\
        0; &\text{if $S=\emptyset$}.
        \end{cases}
    \]
    Then $v(N)=1$.
    The cardinality Banzhaf payouts are
    \[
    \begin{split}
        \cardbanzhaf{1}{v}
        &=
        \frac{3\binom{2}{0}}{2^2}
        \left(
        \frac{1}{\binom{3}{1}}
        \sum_{\substack{S\subseteq N\\ |S|=1}}v(S)
        -
        v(\emptyset)
        \right)
        =
        \frac{3}{4}
        \left( \frac{1+1+1}{3} - 0 \right)
        = \frac{3}{4},\\
        \cardbanzhaf{2}{v}
        &=
        \frac{3\binom{2}{1}}{2^2}
        \left(
        \frac{1}{\binom{3}{2}}
        \sum_{\substack{S\subseteq N\\ |S|=2}}v(S) - \frac{1}{\binom{3}{1}} \sum_{\substack{S\subseteq N\\ |S|=1}}v(S)
        \right)
        =
        \frac{3}{2}
        \left(
        \frac{1+1+1}{3} - \frac{1+1+1}{3}
        \right)
        =
        0,\\
        \cardbanzhaf{3}{v}
        &=
        \frac{3\binom{2}{2}}{2^2}
        \left(
        v(N) - \frac{1}{\binom{3}{2}} \sum_{\substack{S\subseteq N\\ |S|=2}}v(S)
        \right)
        =
        \frac{3}{4}
        \left(
        1 - \frac{1+1+1}{3}
        \right)
        =
        0.
    \end{split}
    \]
    Hence, $\sum_{c=1}^{3}\cardbanzhaf{c}{v} = \frac{3}{4} \neq 1 = v(N)$.
    This contradicts efficiency.
    Hence, $\cardbanzhafempty$ does not satisfy efficiency.
\end{proof}

\cardBanzhafNegativeDummyCardinality*
\begin{proof}
    Consider the game over the feature set $N=\{x_1, x_2, x_3\}$ defined by
    $v(S)=|S|$ for every $S\subseteq N$.
    Then, for cardinality 1, every $x \in N$, $v(\{x\}) = 1 = v(\emptyset) + 1$ (since $v(\emptyset) = 0$).
    Hence, cardinality $1$ is dummy with constant contribution $h_1 = 1$.
    Dummy cardinality therefore requires $\cardbanzhaf{1}{v}=1.$
    However,
    \[
        \begin{split}
        \cardbanzhaf{1}{v}
        &=
        \frac{3\binom{2}{0}}{2^2}
        \left(
        \frac{1}{\binom{3}{1}}
        \sum_{\substack{S\subseteq N \\ |S| = 1}}v(S)
        -
        v(\emptyset)
        \right)
        =
        \frac{3}{4}
        \left(
        \frac{1+1+1}{3}
        -
        0
        \right)
        =
        \frac{3}{4}.
        \end{split}
    \]
    Thus, $\cardbanzhaf{1}{v}\neq h_1$, contradicting dummy cardinality.
    Hence, $\cardbanzhafempty$ does not satisfy dummy cardinality.
\end{proof}

\cardBanzhafPositiveAbsoluteAnonymity*
\begin{proof}
    The claim follows from the fact that $\cardbanzhafempty$ is the $\upsiempty$-value with an added weight (Proposition~2 in~\cite{Upsilon-values2024}).
    For every cardinality $c \in \{1, \ldots, |N|\}$,
    \[
        \cardbanzhaf{c}{v}
        \;=\;
        \frac{|N| \binom{|N|-1}{c-1}}{2^{|N|-1}}
        \left(
        \frac{1}{\binom{|N|}{c}}
        \sum_{\substack{S \subseteq N\\ |S| = c}} v(S) - 
        \frac{1}{\binom{|N|}{c-1}}
        \sum_{\substack{S' \subseteq N\\
        |S'|=c-1}} v(S')
        \right)
        =
        \frac{|N|\binom{|N|-1}{c-1}}{2^{|N|-1}}
        \upsi{c}{v}.
    \]
    The weight depends only on $|N|$ and $c$, and not on the identity of any player.
    Since $\upsiempty$ satisfies absolute anonymity by \Cref{theorem:upsilon-values-standard}, relabelling the players leaves each $\upsiempty$ payout unchanged.
    Multiplying each such payout by the Banzhaf index's identity-independent weight therefore preserves absolute anonymity.
    Hence, $\cardbanzhafempty$ also satisfies absolute anonymity.
\end{proof}

\cardBanzhafNegativeSuccess*
\begin{proof}
    Consider a game over the feature set $\{x_1,  x_2\}$, where $v$ is defined as follows:
    \begin{equation*}
        v(\{x_1\})=1, \qquad v(\{x_2\})=-1, \qquad v(\{x_1, x_2\})=0.
    \end{equation*}
    Calculating the $\cardbanzhafempty$ values yields:
    \begin{equation}
    \label{eq:card-banzhaf-value-calc}
        \begin{split}
            \cardbanzhaf{1}{v}
            &=
            \frac{2\binom{1}{0}}{2}
            \left(
            % \frac{v({x_1}) + v({x_2})}{2}
            \frac{1}{\binom{2}{1}}\left(v(\{x_1\}) + v(\{x_2\})\right)
            -
            v(\emptyset)
            \right)
            =
            0,\\
            \cardbanzhaf{2}{v}
            &=
            \frac{2\binom{1}{1}}{2}
            \left(
            v(\{x_1,x_2\})
            -
            % \frac{v({x_1}) + v({x_2})}{2}
            \frac{1}{\binom{2}{1}}\left(v(\{x_1\}) + v(\{x_2\})\right)
            \right)
            =
            0.
        \end{split}
    \end{equation}
    Success demands that either $\cardbanzhaf{1}{v}$ or $\cardbanzhaf{2}{v}$ should be non-zero, thereby contradicting the above.
    Hence, success is not satisfied by $\cardbanzhafempty$ values.
\end{proof}

\cardBanzhafNegativeCounterfactuality*
\begin{proof}
    Consider the value function $v$ over the feature set $\{x_1, x_2, x_3\}$ defined as follows.
    \[
        v(S):=
        \begin{cases}
           1; &\text{if $S=\{{ x_1}\}$ or $S=\{{ x_1, x_2, x_3}\}$} ,
           \\
           -1; &\text{if $S=\{ x_2\}$},
           \\
           0; &\text{otherwise}.
        \end{cases}
    \]
    Note that
    \[
        \cardbanzhaf{1}{v}
        =
        \frac{3\binom{2}{0}}{2^2}
        \left(
            \frac{1}{\binom{3}{1}}
            \sum_{\substack{S\subseteq N\\ |S|=1}}v(S)
            -
            v(\emptyset)
        \right)
        =
        0.
    \]
    Since $\frac{3\binom{2}{0}}{2^2} = \frac{3}{4} \neq 0$ and $v(\emptyset) = 0$, this means that the average value of singleton coalitions is $0$.
    Hence, if cardinality-based counterfactuality was satisfied, then it would imply $v(N) - \frac{1}{\binom{3}{1}} \sum_{\substack{S\subseteq N \\ |S|=1}}v(S) = 0.$
    However, $v(N) - \frac{1}{\binom{3}{1}} \sum_{\substack{S\subseteq N\\ |S|=1}}v(S) = 1 > 0,$
    contradicting the implication.
    Hence, the cardinality-based counterfactuality principle is not satisfied by $\cardbanzhafempty$.
\end{proof}

%%% CARDINALITY BANZHAF %%%
\cardOwenEfficiency*
\begin{proof}
    By construction, $\cardowen{c}{v}$ collects precisely those terms of the ordinary Owen values for which the coalition formed after the relevant player joins has cardinality $c$ (see \Cref{appendix:derivation-owen-cardinality}).
    Since every Owen term is assigned to exactly one cardinality we have: $\sum_{c=1}^{|N|} \cardowen{c}{v} = \sum_{i \in N} \owen{i}{v}$.
    By \Cref{theorem:owen-classical-principles} the single player Owen value $\owenempty$ is efficient, so $\sum_{i \in N} \owen{i}{v} = v(N)$.
    Hence, $\cardowenempty$ satisfies cardinality-based efficiency.
\end{proof}

\cardOwenDummy*
\begin{proof}
    Let $c$ be a dummy cardinality with constant $h_c$.
    Recall from \Cref{powerindex:cardinality-owen} that $\cardowen{c}{v}$ is a weighted sum of terms $m_{k,R}(t)$, where the weights and the collection of terms depend only on the coalition structure $M_\players$, the cardinality $c$ and the number of players.
    The value function enters only through the terms themselves.

    We first show that every term equals $h_c$.
    The constraint $t = c - |Q_R|$ ensures that each term describes what happens when one player of $S^k$ joins a coalition of size $c-1$, namely the players of $Q_R$ together with $t-1$ teammates from $S^k$.
    Since $c$ is a dummy cardinality, every such step increases the value by exactly $h_c$, no matter which coalition of size $c-1$ we start from or which player joins.
    As $m_{k,R}(t)$ is an average of these steps, it equals $h_c$ as well.
    Writing $W_c$ for the sum of the weights collected at cardinality $c$, we therefore have $\cardowen{c}{v} = W_c \, h_c$ whenever $c$ is a dummy cardinality.

    As noted above, $W_c$ is a fixed number once $M_\players$, $c$ and $|N|$ are fixed.
    Now, let $v^*$ assign $1$ to every coalition of size at least $c$ and $0$ to all others.
    For $v^*$ every cardinality is dummy since the step into size $c$ always adds $1$, while every other step adds $0$.
    By the first part of this proof, $\cardowen{c}{v^*} = W_c$ and $\cardowen{c'}{v^*} = 0$ for all $c' \neq c$.
    By \Cref{theorem:cardinality-owen-value-efficiency} these payouts sum to $v^*(N) = 1$, so $W_c = 1$.

    Since $W_c$ does not depend on the game, we get $\cardowen{c}{v} = W_c \cdot h_c = 1 \cdot h_c = h_c$ for every game in which $c$ is dummy.
\end{proof}

\cardOwenAbsoluteAnonymity*
\begin{proof}
    Let $\pi$ be a permutation of $N$, where renaming the players renames the a priori block structure along with them, so the permuted game is played under the structure $\pi(M_\players) = \{\pi(S^1), \ldots, \pi(S^m)\}$.

    We show that the two payouts are built from exactly the same terms.
    The number of blocks $m$ is unchanged, and each block keeps its size, since $|\pi(S^k)| = |S^k|$.
    The weights in \Cref{powerindex:cardinality-owen} depend only on $m$ and the number of blocks that have joined, so they are unchanged.
    Likewise, the outer coalition determined by $R$ becomes $\pi(Q_R)$, which has the same size as $Q_R$, so the constraints $t = c - |Q_R|$ and $1 \leq t \leq |S^k|$ select the same pairs $(k, R)$ as before.
    Finally, the coalitions entering $m_{k,R}(t)$ are simply renamed such that a coalition $Q_R \cup U$ becomes $\pi(Q_R \cup U)$, and by definition $v_\pi$ assigns it the same value as $v$ assigns to $Q_R \cup U$.
    Hence every term of the permuted sum equals the corresponding term of the original, and the two payouts coincide.
\end{proof}

\cardOwenSuccess*
\begin{proof}
    Consider the game of \Cref{ex:success-failure-upsilon-value} together with the singleton coalition structure $M_N = \{\{x_1\}, \{x_2\}\}$.
    Each block has $|S^k| = 1$, so the constraint $1 \leq t \leq |S^k|$ in \Cref{powerindex:cardinality-owen} forces $t = 1$ and hence $|Q_R| = c - 1$.
    Each term is then an ordinary marginal contribution $v(Q_R \cup \{x_k\}) - v(Q_R)$ with $|Q_R| = c-1$, weighted as in the Shapley value, so $\cardowen{c}{v} = \cardshap{c}{v} = \upsi{c}{v}$ for every $c$.
    By \Cref{ex:success-failure-upsilon-value}, $\upsi{1}{v} = \upsi{2}{v} = 0$, hence $\cardowen{1}{v} = \cardowen{2}{v} = 0$ while $v$ is not the zero function.
    Success demands a cardinality with a non-zero payout, so it is not satisfied by $\cardowenempty$.
\end{proof}

\cardOwenCounterfactuality*
\begin{proof}
    Consider the game of \Cref{ex:counterfactuality-failure-upsilon-value} together with the singleton coalition structure $M_N = \{\{x_1\}, \{x_2\}, \{x_3\}\}$.
    As in the proof of \Cref{theorem:cardinality-owen-value-success}, a singleton structure forces $t = 1$, so $\cardowen{c}{v} = \cardshap{c}{v} = \upsi{c}{v}$ for every $c$.
    By \Cref{ex:counterfactuality-failure-upsilon-value}, $\upsiempty$ violates counterfactuality and quantitative counterfactuality on this game, and the violation carries over to $\cardowenempty$.
\end{proof}

\absoluteanonymitycardcounterfactuality*
\begin{proof}
    Fix an arbitrary game $(N,v)$, and let $\pi$ be an arbitrary but fixed permutation on $N$. 
    For our purposes, it suffices to show that the condition for absolute anonymity, namely $\cardcounterfact{i}{\pi_v}=\cardcounterfact{i}{v}$
    is equivalent to the following: 
    \begin{equation}
    \label{eq:absolute-anonymity-card-counterfactuality}\sum_{|S|=i}v(S)=\sum_{|S|=i} v_{\pi}(S).  
    \end{equation}
    Notice that \Cref{eq:absolute-anonymity-card-counterfactuality} is a summation over every coalition of size $i$.
    Since $|S|=i$  iff $|\hat\pi(S)|=i$, any coalition $S$ appears on the left-hand-side of \Cref{eq:absolute-anonymity-card-counterfactuality} iff $\hat\pi^{-1}(T)$ appears on the right-hand-side.
    Since \Cref{eq:absolute-anonymity-card-counterfactuality} takes the sum over all the possible coalitions of size $i$, both the left-hand-side and right-hand-side are indeed equal.
    Now, the proof follows via the following chain of equivalences:
    \begin{equation*}
       \sum_{|S|=i}v(S)=\sum_{|S|=i} v_{\pi}(S) 
        \iff
       \left[v(N)-\sum_{|S|=i}v(S)\right]=\left[v_{\pi}(N)-\sum_{|S|=i} v_{\pi}(S)\right] 
       \iff
       \cardcounterfact{i}{v}=\cardcounterfact{i}{v_{\pi}}.
    \end{equation*}

    Recall that $v_{\pi}(\hat\pi(S))=v(S)$, or equivalently
    $v_{\pi}(S)=v(\hat\pi^{-1}(S))$.
    Hence, the first equivalence above follows due to the fact that $v_\pi(N)=v_\pi(\hat\pi(N))=v(N)$. 
    The second equivalence follows directly from the definition of the cardinality-based counterfactual power index. 
    This closes the proof.
\end{proof}
%%%%%%%%%%%%%%%%%%%%%%%%%%%%%%%%%%%%%%%%%%%
% END OF PRINCIPLE-BASED ANALYSIS PROOFS
%%%%%%%%%%%%%%%%%%%%%%%%%%%%%%%%%%%%%%%%%%%

%%%%%%%%%%%%%%%%%%%%%%%%%%%%%%%%%%%%%%%%%%%
% BEGINNING DERIVATION OF CARDINALITY INDICES
%%%%%%%%%%%%%%%%%%%%%%%%%%%%%%%%%%%%%%%%%%%
\section{Derivation of cardinality-based versions of power indices}\label{appx:cardinality-derivations}
We perform a size decomposition of the Shapley, Banzhaf, and Owen values to attribute power to cardinalities instead of players.
We do this by considering the original definitions of these power indices, expanding their sums, and reordering the terms.
In doing so, we shift the focus from players to coalition sizes.
In a sense, we consider the game from the perspective in which cardinalities are the ``players''.

\subsection{Cardinality Shapley and Cardinality Banzhaf}
Let $N$ be the set of players, $n = \left|N\right|$ be the number of players, and $i \in N$ a player.
We call $S \subseteq N$ a coalition, and define the shorthand
\begin{equation}
    \Delta_i\left(S\right) = v\left(S \cup \left\{i\right\}\right) - v\left(S\right).
\end{equation}
The Shapley value $\phi$ and Banzhaf value $\beta$, respectively, are defined as
\begin{equation}
    \phi_i = \sum_{S \subseteq N \setminus \left\{i\right\}} \gamma^\text{Sh}\left(s\right) \Delta_i\left(S\right),
    \qquad
    \beta_i = \sum_{S \subseteq N \setminus \left\{i\right\}} \gamma^\text{Bz}\left(s\right) \Delta_i\left(S\right).
\end{equation}
They differ only in their weights,
\begin{equation}\label{eqn:weights}
    \gamma^\text{Sh}\left(s\right) = \frac{s! \left(n-s-1\right)!}{n!} = \frac{1}{n} \binom{n-1}{s}^{-1},
    \qquad
    \gamma^\text{Bz}\left(s\right) = \frac{1}{2^{n-1}}.
\end{equation}
In the above, we have made use of the binomial coefficient's definition,
\begin{equation}\label{eqn:binomial}
    \binom{n}{k} = \frac{n!}{k!\left(n-k\right)!} \quad \Rightarrow \quad \binom{n}{k}^{-1} = \frac{k!\left(n-k\right)!}{n!},
\end{equation}
which will also be useful in the following derivation.

% ==============================================================================
\subsubsection{Group Player $i$'s Terms by Coalition Size}
We begin by grouping the terms of player $i$ by coalition size.
There are $\binom{n-1}{s}$ coalitions of size $s$ that player $i$ could join, which is, so far independent of the power index,
\begin{equation}
    \Delta_i\left(s\right) = \frac{1}{\binom{n-1}{s}} \sum_{\substack{S \subseteq N \setminus \left\{i\right\} \\ \left|S\right| = s}} \Delta_i\left(S\right),
\end{equation}
and, moving the weight term to the other side, we get
\begin{equation}
    \sum_{\substack{S \subseteq N \setminus \left\{i\right\} \\ \left|S\right| = s}} \Delta_i\left(S\right) = \binom{n-1}{s} \Delta_i\left(s\right).
\end{equation}
The common per-player form for the contribution across all coalition sizes is
\begin{equation}\label{eqn:per-player-form}
    \Phi_i = \sum_{s=0}^{n-1} \gamma\left(s\right) \binom{n-1}{s} \Delta_i\left(s\right),
\end{equation}
where $\Phi$ is an arbitrary power index associated with weights $\gamma$.
Now we plug in the Shapley weight $\gamma^\text{Sh}$ and the Banzhaf weight $\gamma^\text{Bz}$, respectively, and get
\begin{equation}\label{eqn:per-player-simplified}
    \phi_i = \sum_{s=0}^{n-1} \frac{1}{n} \Delta_i\left(s\right), \qquad \beta_i = \sum_{s=0}^{n-1} \frac{\binom{n-1}{s}}{2^{n-1}} \Delta_i\left(s\right).
\end{equation}
Here, Shapley's inverse binomial weight cancelled out, but Banzhaf's constant weight factors did not.

% ==============================================================================
\subsubsection{Drop Player Identity}
We drop player identities by summing over all players, and we require conservation---the total power over all players must equal the total power over all cardinalities.
\begin{align}
    \sum_{i\in N}\Phi_i
    & = \sum_{s=0}^{n-1} \gamma(s) \sum_{i\in N} \sum_{\substack{S\subseteq N\setminus \left\{i\right\} \\ \left|S\right|=s}}\Delta_i(S) \\
    & = \sum_{s=0}^{n-1} \gamma(s) \underbrace{\Biggl[ \underbrace{\sum_{i\in N} \sum_{\substack{S\subseteq N\setminus \left\{i\right\} \\ \left|S\right|=s}} v\left(S \cup \left\{i\right\}\right)}_{A\left(s\right)} - \underbrace{\sum_{i\in N} \sum_{\substack{S\subseteq N\setminus \left\{i\right\} \\ \left|S\right|=s}} v\left(S\right)}_{B\left(s\right)} \Biggr]}_{C\left(s\right)}
\end{align}
What are $A$ and $B$?
$A$ contains those cases where player $i$ has joined the coalition,
\begin{equation}
    A\left(s\right) = \sum_{i\in N} \sum_{\substack{S\subseteq N\setminus \left\{i\right\} \\ \left|S\right|=s}} v\left(S \cup \left\{i\right\}\right)
    = \left(s+1\right) \sum_{\left|S'\right| = s+1} v\left(S'\right)
    = \left(s+1\right) \binom{n}{s+1} \overline{v}\left(s+1\right)
\end{equation}
And $B$ contains those cases where player $i$ has not joined the coalition,
\begin{equation}
    B\left(s\right) = \sum_{i\in N} \sum_{\substack{S\subseteq N\setminus \left\{i\right\} \\ \left|S\right|=s}} v\left(S\right)
    = \left(n-s\right) \sum_{\left|S\right| = s} v\left(S\right)
    = \left(n-s\right) \binom{n}{s} \overline{v}\left(s\right).
\end{equation}
Here, $\overline{v}\left(s\right)$ is the average value of coalitions of size $s$,
\begin{equation}
    \overline{v}\left(s\right) = \frac{1}{\binom{n}{s}} \sum_{\substack{S \subseteq N \\ \left|S\right| = s}} v\left(S\right).
\end{equation}
Now we put $A$ and $B$ together to get $C$,
\begin{equation}
    C\left(s\right) = A\left(s\right) - B\left(s\right) = \left(s+1\right) \binom{n}{s+1} \overline{v}\left(s+1\right) - \left(n-s\right) \binom{n}{s} \overline{v}\left(s\right),
\end{equation}
where $s \in \left[0,n-1\right]$.
Because $\left(s+1\right)\binom{n}{s+1} = n \binom{n-1}{s}$ and $\left(n-s\right) \binom{n}{s} = n \binom{n-1}{s}$, we have
\begin{equation}
    C\left(s\right) = n \binom{n-1}{s} \overline{v}\left(s+1\right) - n \binom{n-1}{s} \overline{v}\left(s\right) = n \binom{n-1}{s} \underbrace{\left[\overline{v}\left(s+1\right) - \overline{v}\left(s\right) \right]}_{=: m\left(s+1\right)}.
\end{equation}
Here, $m\left(s+1\right)$ is the difference in average values between size-$\left(s+1\right)$ and size-$s$ coalitions.
Together, we have
\begin{equation}
    \sum_i \Phi_i = \sum_{s=0}^{n-1} \gamma\left(s\right) C\left(s\right) = \sum_{s=0}^{n-1} \gamma\left(s\right) n \binom{n-1}{s} m\left(s+1\right) = \sum_{s=1}^n \gamma\left(s-1\right) n \binom{n-1}{s-1} m\left(s\right).
\end{equation}

% ==============================================================================
\subsubsection{Plug in Shapley and Banzhaf Weights}
We plug in the Shapley and Banzhaf weights from \Cref{eqn:weights},
\begin{equation}
    \sum_i \phi_i = \sum_{s=0}^{n-1} m\left(s+1\right) = \sum_{s=1}^{n} m\left(s\right) \qquad \sum_i \beta_i = \sum_{s=0}^{n-1} \frac{n \binom{n-1}{s}}{2^{n-1}} m\left(s+1\right) = \sum_{s=1}^{n} \frac{n \binom{n-1}{s-1}}{2^{n-1}} m\left(s\right).
\end{equation}
For Shapley, the weights have completely cancelled, but for Banzhaf, they have not.
The cardinality-based versions of Shapley and Banzhaf are
\begin{equation}
    \phi_s = m\left(s\right), \qquad \beta_s = \frac{n \binom{n-1}{s-1}}{2^{n-1}} m\left(s\right), \qquad s = 1,\ldots,n.
\end{equation}

% ==============================================================================
\subsubsection{Example}
We provide an example where cardinality Shapley and Banzhaf return different ranking of cardinalities.
Let
\begin{equation}
    N = \left\{X,Y,Z\right\},
    \qquad 
    v\left(S\right) = \begin{cases}
      0 & \left|S\right| = 0 \\
      3 & \left|S\right| = 1 \\
      5 & \left|S\right| = 2 \\
      6 & \left|S\right| = 3
    \end{cases}
\end{equation}
Because $v\left(S\right)$ only depends on the size $s$ of $S$, we have $\overline{v}\left(s\right) = v\left(S\right)$.
Now plug the cardinalities into $m$,
\begin{equation}
    m\left(s\right) = \overline{v}\left(s\right) - \overline{v}\left(s-1\right),
    \qquad
    m\left(1\right) = 3,
    \qquad
    m\left(2\right) = 2,
    \qquad
    m\left(3\right) = 1.
\end{equation}
This game is a game with diminishing returns, where adding one more player adds less value as the cardinality increases.
Now we compute the weights for Shapley and Banzhaf,
\begin{equation}
    w_S = \left(1, 1, 1\right), \qquad w_B\left(s\right) = \frac{n \binom{n-1}{s-1}}{2^{n-1}} = \frac{3 \binom{2}{s-1}}{2^2} = \left(\frac{3}{4}, \frac{6}{4}, \frac{3}{4}\right).
\end{equation}
And, putting things together, we get
\begin{equation}
    \phi_1 = 3, \qquad \phi_2 = 2, \qquad \phi_3 = 1, \qquad \beta_1 = \frac{9}{4}, \qquad \beta_2 = \frac{12}{4}, \qquad \beta_3 = \frac{3}{4},
\end{equation}
corresponding to the rankings of cardinalities
\begin{equation}
    \text{Shapley:~} 1 > 2 > 3 \qquad \text{Banzhaf:~} 2 > 1 > 3
\end{equation}

%%%%%%%%%%%%%%%%%%
% OWEN CARDINALITY
%%%%%%%%%%%%%%%%%%
\subsection{Cardinality Owen}\label{appendix:derivation-owen-cardinality}
Recall that the Owen value models the situation in which a priori coalition structures hold, that is, certain coalitions would never form in practice.
Let $M_\players = \left\{S^1, \ldots, S^m\right\}$ be a partition of the players $N$ that defines the coalition structure, where $\mathcal{I} = \left\{1, \ldots m\right\}$ are indices into $M_\players$.
For each subset $R \subseteq \mathcal{I}$, we define $Q_R = \bigcup_{j \in R} S^j$, that is, $Q_R$ is the set of players in the blocs indexed by $R$.
With player $i \in S^k$, the Owen value is defined as
\begin{equation}
    \omega_i = \sum_{R \subseteq \mathcal{I} \setminus \left\{k\right\}} \sum_{T \subseteq S^k \setminus \left\{i\right\}} \gamma^\text{Ow}\left(r,t\right) \Delta_i \left( Q_R \cup T \right),
\end{equation}
which can also be written as
\begin{equation}
    \omega_i = \sum_{r=0}^{m-1} \sum_{t=0}^{s_k-1} \gamma^\text{Ow}\left(r,t\right) \sum_{\substack{R \subseteq \mathcal{I} \setminus \left\{k\right\} \\ \left|R\right| = r}} \sum_{\substack{T \subseteq S^k \setminus \left\{i\right\} \\ \left|T\right| = t}} \Delta_i\left(Q_R \cup T\right),
\end{equation}
where $s_k = \left|S^k\right|$.
That is, we replace the outer sums to run over the number of outside blocs $r$ and the number of inside-bloc players $t$---from the perspective of player $i$---and introduce inner sums where we pick out the outside and inside coalitions of those sizes $r$ and $t$.
The Owen value follows the idea that blocs in the ``outer game'' join ``all or nothing'', and players in the ``inner game'' join one player at a time.

The Owen weight is
\begin{equation}\label{eqn:owen-weights}
    \gamma^\text{Ow}\left(r,t\right) = \frac{r! \left(m-1-r\right)!}{m!} \cdot \frac{t! \left(s_k-1-t\right)!}{s_k!} = \frac{1}{m} \binom{m-1}{r}^{-1} \frac{1}{s_k} \binom{s_k-1}{t}^{-1},
\end{equation}
where $r = \left|R\right|$, $t = \left|T\right|$, consists of two Shapley weights and is the probability that players $Q_R \cup T$ have joined before $i$.
In the case of a trivial partition $M_\players = \left\{N\right\}$ or $M_\players = \left\{\left\{i\right\} \,|\, i \in N\right\}$, the Owen value becomes identical with the Shapley value because the ``outer game'' collapses.

%We also define the quotient game $\left(M, v^\mathcal{P}\right)$, where each bloc $P_i \in \mathcal{P}$ acts like a player, and
%
% \begin{equation}
%     v^\mathcal{P}\left(R\right) = v\left(Q_R\right), \qquad R \subseteq M.
% \end{equation}

% ------------------------------------------------------------------------------
\subsubsection{Group Player $i$'s Terms by Coalition Size}
Different from before, we need to consider coalition formation in two stages, that is, in the ``outer game'', and the ``inner game''.
We denote the number of blocs that have joined in the outer game by $r$, and the number of players who have joined in the inner game as $t$.
However, we are aiming for a cardinality-based version of Owen that considers only a single size dimension---the size of the coalition at the moment when a player $i$ joins.
We fix the outer coalition by fixing $R$,
\begin{equation}
    \Delta_i\left(R,t\right) = \frac{1}{\binom{s_k-1}{t}} \sum_{\substack{T \subseteq S^k \setminus \left\{i\right\} \\ \left|T\right| = t}} \Delta_i\left(Q_R \cup T\right).
\end{equation}
These are, for player $i$ and with ``outside union'' $R$ fixed, the possible ``inside unions'' of size $t$ that $i$ can join---there are $\binom{s_k-1}{t}$ many.
Moving the weight to the other side yields
\begin{equation}
    \sum_{\substack{T \subseteq S^k \setminus \left\{i\right\} \\ \left|T\right| = t}} \Delta_i\left(Q_R \cup T\right) = \binom{s_k-1}{t} \Delta_i\left(R,t\right).
\end{equation}
In analogy with \Cref{eqn:per-player-form}, this gives us the per-player power
\begin{equation}
    \omega_i = \sum_{r=0}^{m-1} \sum_{\substack{R \subseteq \mathcal{I} \setminus \left\{k\right\} \\ \left|R\right| = r}} \sum_{t=0}^{s_k-1} \gamma^\text{Ow}\left(r,t\right) \binom{s_k-1}{t} \Delta_i\left(R,t\right).
\end{equation}

% ------------------------------------------------------------------------------
\subsubsection{Drop Player Identity}
We sum over all players to get rid of player identities,
\begin{align}
    \sum_{i \in N} \omega_i & = \sum_{k} \sum_{i \in S^k} \omega_i \\
    & = \sum_k \sum_{i \in S^k} \sum_{r=0}^{m-1} \sum_{\substack{R \subseteq \mathcal{I} \setminus \left\{k\right\} \\ \left|R\right| = r}} \sum_{t=0}^{s_k-1} \gamma^\text{Ow}\left(r,t\right) \binom{s_k - 1}{t} \Delta_i\left(R,t\right) \\
    & = \sum_k \sum_{r=0}^{m-1} \sum_{\substack{R \subseteq \mathcal{I} \setminus \left\{k\right\} \\ \left|R\right| = r}} \sum_{t=0}^{s_k-1} \gamma^\text{Ow}\left(r,t\right) \sum_{i \in S^k} \sum_{\substack{T \subseteq S^k \setminus \left\{i\right\} \\ \left|T\right| = t}} \Delta_i\left(Q_R \cup T\right) \\
    & = \sum_k \sum_{r=0}^{m-1} \sum_{\substack{R \subseteq \mathcal{I} \setminus \left\{k\right\} \\ \left|R\right| = r}} \sum_{t=0}^{s_k-1} \gamma^\text{Ow}\left(r,t\right) \underbrace{\Biggl[ \underbrace{\sum_{i \in S^k} \sum_{\substack{T \subseteq S^k \setminus \left\{i\right\} \\ \left|T\right| = t}} v\left(Q_R \cup T \cup \left\{i\right\} \right)}_{A_{k,R}\left(t\right)} - \underbrace{\sum_{i \in S^k} \sum_{\substack{T \subseteq S^k \setminus \left\{i\right\} \\ \left|T\right| = t}} v\left(Q_R \cup T \right)}_{B_{k,R}\left(t\right)} \Biggr]}_{C_{k,R}\left(t\right)}
\end{align}
As before, $A$ captures the cases where $i$ has joined the coalition, and $B$ captures the cases where $i$ has not joined.
The important detail here is that we have fixed the outer coalition $R$ rather than merely counting how many blocs are present in the outer game.
Why?
Because, if we only tracked the number of outer blocs, $r$, a single $r$ could form coalitions of different sizes $\left|Q_R\right|$ if the only constraint is $\left|R\right| = r$.
We now consider and simplify $A$ and $B$ in turn.
\begin{equation}
    A_{k,R}\left(t\right)
    = \sum_{i \in S^k} \sum_{\substack{T \subseteq S^k \setminus \left\{i\right\} \\ \left|T\right| = t}} v\left(Q_R \cup T \cup \left\{i\right\} \right)
    = \left(t+1\right) \binom{s_k}{t+1} \overline{v}_{k,R}\left(t+1\right)
\end{equation}
\begin{equation}
    B_{k,R}\left(t\right) 
    = \sum_{i \in S^k} \sum_{\substack{T \subseteq S^k \setminus \left\{i\right\} \\ \left|T\right| = t}} v\left(Q_R \cup T \right)
    = \left(s_k-t\right) \binom{s_k}{t} \overline{v}_{k,R}\left(t\right)
\end{equation}
Here, as before, $\overline{v}$ is the average value of a coalition with a specific size, but now we fix the outer coalition $R$,
\begin{equation}
    \overline{v}_{k,R}\left(\tau\right) = \frac{1}{\binom{s_k}{\tau}} \sum_{\substack{U \subseteq S^k \\ \left|U\right| = \tau}} v\left(Q_R \cup U\right).
\end{equation}
Now, via the definition of binomial coefficients, we use the fact that $\left(t+1\right) \binom{s_k}{t+1} = \left(s_k-t\right) \binom{s_k}{t} = s_k \binom{s_k-1}{t}$, and put $A$ and $B$ together to get $C$,
\begin{equation}
    C_{k,R}\left(t\right) = A_{k,R}\left(t\right) - B_{k,R}\left(t\right) = s_k \binom{s_k-1}{t} \underbrace{\left[\overline{v}_{k,R}\left(t+1\right) - \overline{v}_{k,R}\left(t\right)\right]}_{=: m_{k,R}\left(t+1\right)}
\end{equation}

% ------------------------------------------------------------------------------
\subsubsection{Plug in Owen Weight}
The binomial and inverse binomial coefficients cancel again,
\begin{align}
    \gamma^\text{Ow} \left(r,t\right) C_{k,R}\left(t\right) 
    & = \frac{1}{m} \binom{m-1}{r}^{-1} \frac{1}{s_k} \binom{s_k-1}{t}^{-1} \cdot s_k \binom{s_k-1}{t} m_{k,R}\left(t+1\right) \\
    & = \frac{r!\left(m-1-r\right)!}{m!} m_{k,R}\left(t+1\right).
\end{align}
Now we reindex $t \to t-1$ and get
\begin{equation}
    \sum_{i \in N} \omega_i = \sum_{k=1}^{m} \sum_{r=0}^{m-1} \frac{r!\left(m-1-r\right)!}{m!} \sum_{\substack{R \subseteq \mathcal{I} \setminus \left\{k\right\} \\ \left|R\right| = r}} \sum_{t=1}^{s_k} m_{k,R}\left(t\right).
\end{equation}

% ------------------------------------------------------------------------------
\subsubsection{Make the size explicit}
We can interpret $m_{k,R}\left(t\right)$ as the ``moment'' when a player of $S^k$ joins $\left(t-1\right)$ teammates who are already present, producing a coalition of size $c = \left|Q_R\right| + t$, where $\left|Q_R\right| = \sum_{j \in R} |S^j|$.
Now, we collect all terms whose resulting coalition has size $c$.
For each $\left(k,R\right)$, the summation constraints ensure that $t = c - \left|Q_R\right|$ takes on only one value, and every term is counted exactly once.
We get the cardinality version of Owen,
%
% \begin{equation}
%     \omega_\ell = \sum_{k=1}^m \sum_{r=0}^{m-1} \frac{r!\left(m-1-r\right)!}{m!} \sum_{\substack{R \subseteq M \setminus \left\{k\right\}, \left|R\right| = r \\ t = \ell - \left|Q_R\right|, 1 \leq t \leq p_k}} m_{k,R}\left(t\right), \qquad \ell = 1,\ldots,n.
% \end{equation}
\begin{equation}
    \cardowen{c}{v} = \sum_{k=1}^m \sum_{r=0}^{m-1} \frac{r!\left(m-1-r\right)!}{m!} \sum_{\substack{R \subseteq \mathcal{I} \setminus \left\{k\right\}, \left|R\right| = r \\ t = c - \left|Q_R\right|, 1 \leq t \leq |S^k|}} m_{k,R}\left(t\right), \qquad c = 1,\ldots,n.
\end{equation}

% ------------------------------------------------------------------------------
\subsubsection{Example}
Cardinality Owen and Shapley can rank cardinalities differently when there is a meaningful bloc structure.
The reason is that Cardinality Shapley assumes naively that all coalitions can form in any order.
From the perspective of orderings, Shapley considers all orderings to be possible, whereas Owen only allows those that are consistent with the assumption that blocs in the ``outer game'' join as a unit.
Let
\begin{equation}
    N = \left\{X,Y,Z\right\}, \quad M_\players = \left\{\left\{X,Y\right\}, \left\{Z\right\}\right\} \qquad \Rightarrow \qquad m=2, \quad s_1 = 2, \quad s_2 = 1.
\end{equation}
And the value function
\begin{equation}
    v\left(\emptyset\right)
  = v\left(\left\{X\right\}\right) 
  = v\left(\left\{Y\right\}\right) 
  = v\left(\left\{Z\right\}\right)
  = v\left(\left\{X,Y\right\}\right) 
  = 0 
  \quad 
    v\left(\left\{X,Z\right\}\right) 
  = v\left(\left\{Y,Z\right\}\right) 
  = v\left(N\right) = 6.
\end{equation}
Then we get the results presented in Table~\ref{table:card-owen-shap-comparison}.
\begin{table}[h!]
    \centering
    \begin{tabular}{c|cc}
        \toprule
        $c$ & $\omega_c$ & $\phi_c = \overline{v}\left(c\right) - \overline{v}\left(c-1\right)$ \\
        \midrule
        $1$ & $0$ & $0$ \\
        $2$ & $3$ & $4$ \\
        $3$ & $3$ & $2$ \\
        \bottomrule
    \end{tabular}
    \caption{Cardinality-based Owen and Shapley values for the game.}
    \label{table:card-owen-shap-comparison}
    \vspace*{-\baselineskip}
\end{table}

Why do Cardinality Owen and Cardinality Shapley disagree?
Because Shapley naively assumes that coalitions can form in any order, while Owen respects the a priori bloc structure.
This becomes clear when we view coalition formation from the perspective of random orderings.
Here are all possible random orderings of $N$, where those variants that are impermissible according to $M_\players$ are crossed out,
\begin{equation}
    \left[X,Y,Z\right], \quad
    \cancel{\left[X,Z,Y\right]}, \quad
    \left[Y,X,Z\right], \quad
    \cancel{\left[Y,Z,X\right]}, \quad
    \left[Z,X,Y\right], \quad
    \left[Z,Y,X\right].
\end{equation}
This affects how coalitions can be formed: it is possible to form the coalition $S = \left\{X,Z\right\}$ if $Z$'s bloc joins in the outer game, and then $X$ in the inner game.
It is, however, not possible that $X$ joins in the outer game alone for $Z$ to join in the inner game.
Therefore, Shapley permits more ways to form coalitions of size $2$ than Owen.

%%%%%%% END %%%%%%%
\end{document}